\documentclass[%
reprint,
nofootinbib,
amsmath,amssymb,
aps,
amsart
floatfix,
]{revtex4-1}

\usepackage{graphicx}
\usepackage{dcolumn}
\usepackage{bm}
\usepackage{amsmath ,amssymb ,amsfonts,mathrsfs,wasysym,amsthm, textcomp, mathtools,graphicx,stmaryrd,lipsum,amscd}
\DeclareMathAlphabet{\mathpzc}{OT1}{pzc}{m}{it}
\usepackage{hyperref}

\hypersetup{%
	colorlinks=false,
	pdfborderstyle={/S/U/W 2}
}

\newcommand{\Rea}{\mathbb{R}}
\newcommand{\Comp}{\mathbb{C}}
\newcommand{\Id}{\hat{\mathbb{I}}}
\newcommand{\Nat}{\mathbb{N}}

\newcommand{\Hi}{\mathcal{H}}

\newcommand{\borel}{\mathscr{B}(\mathbb{R})}

\newcommand{\E}{\mathcal{E}}

\newcommand{\Ex}{\mathbb{E}}

\newcommand{\Tr}[1]{\mbox{Tr}\left[#1\right]}

\newcommand{\PS}{(\Omega,\E,P)}

\newcommand{\Spa}{\mathbb{S}^A}
\newcommand{\sss}{\mathbf{S}}
\newcommand{\Int}{\mathbb{Z}}
\newcommand{\Spb}{\mathbb{S}^B}
\newcommand{\F}{\mathcal{F}}

\begin{document}
	
	\newtheorem{definition}{Definition}
	\newtheorem{proposition}{Proposition}
	\newtheorem{theorem}{Theorem}
	\newtheorem{lemma}{Lemma}
	\newtheorem{corollary}{Corollary}
	\newtheorem{assumptions}{Assumptions}
	\newtheorem{assumption}{Assumption}
	
	\renewcommand{\thefootnote}{\fnsymbol{footnote}}
	
	\preprint{APS/123-QED}
	
	\title{On non-commutativity in quantum theory (II): \\ toy models for non-commutative kinematics.}
	
	\author{Curcuraci Luca}
	\affiliation{Department of Physics, University of Trieste, Strada Costiera 11 34151, Trieste, Italy \\ Istituto Nazionale di Fisica Nucleare, Trieste Section, Via Valerio 2 34127, Trieste, Italy}
	\email{ Curcuraci.article@protonmail.com; \\
		luca.curcuraci@phd.units.it }
	
	\date{\today}
	
	\begin{abstract}
		
		In this article, we continue our investigation on the role of non-commutativity in quantum theory. Using the method explained in \emph{On non-commutativity in quantum theory (I): from classical to quantum probability}, we analyze two toy models which exhibit non-commutativity between the corresponding position and velocity random variables. In particular, using ordinary probability theory, we study the kinematics of a point-like particle jumping at random over a discrete random space. We show that, after the removal of the random space from the model, the position and velocity of the particle do not commute, when represented as operators on the same Hilbert space. 
		
	\end{abstract}
	
	\maketitle
	
	\tableofcontents
	
	\section{Introduction}
	
	In \cite{LC} we proposed a method to construct a non-commutative probability theory starting from a collection of ordinary probability spaces (i.e. a contextual probability space \cite{khrennikov2009contextual}) using entropic uncertainty relations. In this article, following this method, we try to shed some light on the non-commutativity in quantum mechanics. In particular, we will focus on the fundamental commutation relation
	\begin{equation}\label{[Q,P]}
	[\hat{Q},\hat{P}] = i\hbar \Id
	\end{equation}
	between position and momentum operators in non-relativistic quantum mechanics. The goal is to construct a Kolmogorov probabilistic model in which this commutation relation and, more generally, non-relativistic quantum mechanics can be recovered. To achieve this, we will present two toy models describing the kinematics of a point-like particle jumping at random over a random discrete space. As we will see, in these toy models a full derivation of the commutation relation \eqref{[Q,P]} is not possible (hence no direct comparison with non-relativistic quantum mechanics is available) but from them, we may understand how such a model should look like. Such model will be presented in \cite{LC3}.
	
	The main idea is the following. In the framework of non-relativistic quantum mechanics, the statistical description of a free particle in $\Rea^3$ is done by using an element of a separable infinite-dimensional Hilbert space, $\psi \in L_2(\Rea^3)$, and a set of non-commuting (in general) self-adjoint operators on this Hilbert space. As we have seen, non-commutativity has various consequences like Heisenberg uncertainty principle, CHSH inequalities (when also the spin is considered) and, most important, the impossibility to abandon the Hilbert space description\footnote{A notable exception is the phase-space formulation of  non-relativistic quantum mechanics \cite{fairlie1964formulation,baker1958formulation} which does not use Hilbert spaces. However, it uses quasi-probability distributions where negative probabilities are difficult to understand from the statistical point of view.}. The basic assumption of the models presented here is that space (time is still a parameter) plays an active role in the description of a particle. More precisely, we will treat particles as point-like objects and the \emph{ physical space} as a random distribution of points. With the term \textquotedblleft physical space", we mean the space on which particles actually move: for example the physical space of classical mechanics is $\Rea^3$. The random distribution of points used to describe the physical space is not static but evolves stochastically in time according to some law. We also assume that a particle moves by jumping at random from one point of the physical space to another. The particle and physical space are described using random variables in the framework of ordinary probability theory. When we want to describe \emph{only} the particle, we have to remove (the exact meaning of this term will be clarified later) the random variables describing the physical space: this will be the origin of the non-commutativity between the position and the velocity operators of the particle in these models.
	
	The article is organized as follows. In section \ref{SpaceTime} we will give some physical arguments supporting the basic assumptions of the model about the physical space, then in section \ref{mA}, we will discuss a toy model where time is discrete. This will be generalized in section \ref{mB} to the continuous time case. In both models, we will derive an entropic uncertainty relation for the position and velocity random variables. This allows to conclude that they can be represented on a common Hilbert space as two non-commuting operators (using the results presented in \cite{LC}). For each toy model, we point out positive aspects and limitations.    
	
	\section{Space and time in quantum mechanics}\label{SpaceTime}
	
	In the models proposed in the subsequent sections, space will be treated as a stochastic process, while time will be a parameter. Here we try to argue our choice of space and time using ordinary non-relativistic quantum mechanics.\newline
	
	Let us start with the \emph{time}. In ordinary quantum mechanics, time is a parameter, and we will treat it in the same way also in the proposed models. It is known that, associate to time an operator $\hat{T}$ which is the canonical conjugate of the Hamiltonian $\hat{H}$, i.e. $[\hat{T},\hat{H}] = i \Id$, is problematic \cite{bunge1970so}. Different proposals are available \cite{wang2007introduce} however, none of them can be considered as a satisfactory solution of the problem: 1) one may use an operator $\hat{T}$ which is not self-adjoint  and fulfil the commutation relation, but then one has to deal with complex eigenvalues of such operator;  2) one may choose an Hamiltonian which is not bounded from below and fulfil the commutation relation using a self-adjoint $\hat{T}$, but such Hamiltonian does not describe stable physical systems. Giving up to fulfil the relation $[\hat{T},\hat{H}] = i \Id$, another possible way to introduce a time operator is the following. Suppose we have a quantum particle, described at time $t$ by the vector $|\psi_t \rangle \in \Hi$, and whose time evolution is given by the Sch\"odingher equation, as usual. We want to define the time operator $\hat{T}$ as the operator such that
	\begin{equation*}
	\hat{T} | \psi_t \rangle = t | \psi_t \rangle,
	\end{equation*}
	for any $|\psi_t\rangle \in \Hi$. Since $t \in \Rea$ the spectrum is real, hence $\hat{T}$ is self-adjoint. No commutation relation with the Hamiltonian is assumed, hence we are free to assume the energy spectrum bounded from below. In addition, we also assume that $\hat{T}$ commutes with all the operators over $\Hi$. This is reasonable from the physical point of view since we can always measure time together with any other observable of the particle in a non-relativistic experiment. Indeed in non-relativistic systems, the time in any clock of the laboratory is the time at which the quantum particle is measured in the same experiment. Assuming that, the spectral representation theorem \cite{LC,moretti2013spectral} implies that
	\begin{equation*}
	\Hi = \int^{\bigoplus}_{\Rea} \Hi_t dt
	\end{equation*}
	where $\int^{\oplus} \cdot$ means the continuous direct sum. The unitary time evolution $\hat{U}_t$ induced by the Sch\"{o}dinger equation, can be seen as a map between different Hilbert spaces in the direct sum above, namely $\hat{U}_{s}:\Hi_t  \rightarrow \Hi_{t+s}$. By the spectral decomposition theorem \cite{LC,moretti2013spectral}, the operator $\hat{T}$ can be written as
	\begin{equation*}
	\hat{T} = \int^{\bigoplus}_\Rea t \Id_{\Hi_t} dt
	\end{equation*}
	where $\Id_{\Hi_t}$ is the identity on the Hilbert space $\Hi_t$. In \cite{LC} we saw that non-commuting operators over a Hilbert space are the non-commutative version of random variables and that, their probability distributions are all encoded in a state defined on the algebra that they form (typically represented using a vector of the Hilbert space on which they are defined). In all attempts seen above to define a time operator, we cannot consider time as a random variable with probability distribution \emph{induced by the quantum state used to describe the particle}. Indeed, if $\hat{T}$ is not self-adjoint, it corresponds to a random variable taking value on $\Comp$, which is hardly identifiable with physical time. If $\hat{T}$ is self-adjoint but $\hat{H}$ is not bounded from below, the time is a random variable but of an unphysical system. Finally, in the last possibility, we can easily understand that no statistical information about time is contained in $|\psi_t\rangle$. For these reasons in the proposed models, we can safely treat time as a parameter \textquotedblleft without neglecting possible quantum effects".\newline
	
	Now we turn our attention to \emph{space}. In this case the situation is different. Consider a quantum particle in $\Rea^3$, hence with Hilbert space $L_2(\Rea^3)$. Let $\mathbf{x} = (x_1,x_2,x_3) \in \Rea^3$, the position operator is defined as
	\begin{equation*}
	\hat{Q}_i\psi(\mathbf{x}) = x_i \psi(\mathbf{x}) \mspace{30mu} i= 1,2,3
	\end{equation*}
	for all $\psi(\mathbf{x}) \in \mathcal{S}(\Rea^3)$, i.e. all the Schwartz functions on $\Rea^3$. $\hat{Q}_i$ is self-adjoint and does not commute with the momentum operator. The previous arguments does not apply and it can be legitimately considered as a random variable whose statistical properties are described by the wave function. However, $\hat{Q}_i$ represents on the Hilbert space the random variable describing the $i$-th coordinate of the \emph{particle position}, and is not related to the underlying physical space. The particle position is the random phenomena, not the physical space where the particle lives.
	In appendix A, a simple model of ruler described within the formalism of non-relativistic quantum mechanics is given. In a nutshell, a quantum ruler can be considered as a collection of quantum particles bounded together and localised in a given region of space. Particles are assumed distinguishable, so they can be counted, and each particle can be found in two different states, labeled by a spin variable. Before any measurement, the spin variables of the ruler are in a known configuration. The measurement is modelled with a contact interaction (between the ruler and the particle we want to measure) which generates a spin-flip. A (projective) measurement of the quantum ruler (as a photograph) right after the interaction reveals which particle of the ruler \textquotedblleft touches" the measured quantum particle. We can then count the number of particles between the spin-flipped particle and a chosen origin on the ruler (see the distance functions in appendix B for some possible methods). Repeating this procedure many times, we obtain that the probability to find the $i$-th particle of the ruler with the spin flipped is well approximated by
	\begin{equation*}
	P[X_A = i] = \int \prod_{j = 1}^N dy_j  |\psi_R(y_1,\cdots,y_i, \cdots, y_N)|^2 |\phi_A(y_i)|^2,
	\end{equation*}
	where $N$ is the number of particles of the quantum ruler, $\psi_R(y_1,\cdots,y_i, \cdots, y_N)$ is the wave function of the quantum ruler and $\phi_A(x)$ is the wave function of the quantum particle whose position is measured. In appendix A, it is argued that under reasonable assumptions, the expression above reduces to $|\phi_A(x)|^2$ as expected. Note that in the above expression, we have two contributions to the probability: one due to the particle and one due to the ruler. Hence, if we construct the physical space of a quantum system using a quantum mechanical model of a ruler, we may legitimately think that the physical space of quantum mechanics can be described by random variables.\newline
	
	We conclude by observing that the argument presented here about space \emph{is not loophole free}. One can always argue that the stochasticity we observe in the physical space is an artifact of the ruler: the ruler is random, not the space. This is clearly another legitimate possibility, but in this article, we want to explore the consequences of \emph{the choice of considering space as a random phenomenon}. 
	
	\section{Model A: Discrete-time $1$-D kinematics on a random space}\label{mA}
	
	Here we will describe a discrete (and finite) random space and a particle moving on it jumping at random from one point of space to another. Space, position, and velocity of the particle at a given time will be treated in the same way: using random variables. The whole model is 1-dimensional. We will show that, once the space process is removed from the model, the position and velocity of the particle can be jointly described in a non-commutative probability space.
	
	\subsection{The space process}\label{ModA:SPsec}\label{sec3a}
	
	The process describing space in this model (Model A) we will be called \emph{space process}. The space process is assumed to be a discrete and finite set of points distributed at random. More precisely, at each instant of time,  space is a random distribution of $M \in \Nat$ points over the real line. Such points evolve in time as discrete-time random walks and, in this sense, space is a \emph{stochastic process}. This time evolution has a twofold interpretation. A first possibility is to think it with respect to the real line: a point of the space process is a random walk and it changes its position along the line as time changes. A second possible way to see this time evolution is to look at its effects on the \textquotedblleft ordering among points": the points change their distances with respect to a chosen point (the origin) when this distance is \textquotedblleft measured on the points" (see the distances defined in appendix B). In some sense this second point of view can be considered as an internal description: it describes space as if the observer has no possibility to see the continuous real line. On the other hand, the first possibility should be considered as an external description\footnote{An interesting analogy can be made between the two possibilities explained here for the description of the space and the description of a manifold. A manifold can be studied using a coordinate system on it (internal description), or imagine that is embedded in a larger space (external description), similarly to what happens here.}. For simplicity, we chose to describe the whole model from the first point of view.  Nothing forbids to adopt the second point of view for the description despite, at a first look, it seems more complicated.\newline
	
	Let us recall some basic facts about the random walk \cite{rudnick2004elements}. Consider a lattice of points having spacing $l \in \Rea$, say $\Int_l := \{ x \in \Rea \mspace{5mu} | \mspace{5mu} x = ln, n \in \Int \}$. Then take a collection of independent, identically distributed Bernulli random variables $\{Y_i\}_{i = 1}^\infty$, characterised by the probabilities $P[Y_i = -l] = p$ and $P[Y_i = +l] = q =1-p$ for all $i$. Using this collection, we can define the random walk as the process
	\begin{equation}\label{RandomWalk}
	S_N := \sum_{i =0}^N Y_i
	\end{equation}
	where $Y_0$ is an arbitrary random variable with distribution $\pi(y_0)$ taking value on $\Int_l$, representing the initial position of the random walk. $N$ labels time (assumed discrete) and $S_N$ represents the position of the random walk at time $N$. Let us now derive the probability distribution of the random walk position at time $N$, i.e. $S_N$. Consider the random walk at time $N$. Since at each time-step the random walk can move by $+l$ or by $-l$ its position, if for $n < N$ times the random walk moves by $-l$, its final position $d$ will be
	\begin{equation*}
	d = (N-n)l -nl = (N - 2n)l
	\end{equation*}
	Using this equation we can see that, if at time $N$ the random walk is found in $d$, the number of times the random walk moves by $-l$ is
	\begin{equation*}
	n = \frac{1}{2}\left( \frac{d}{l} + N \right).
	\end{equation*}
	Clearly, the number of times it moves by $+l$ will be $N -n$. Note that the chronological order of the movements does not make any difference on the final position. Assume, for the moment, that the initial position $Y_0$ is given, and set it $Y_0 = 0$. Since for a given $d = (N - 2n)l$ the random walk is just the sum of Bernulli random variables, i.e. a binomial process, we can write
	\begin{equation*}
	\begin{split}
	P[S_N = d] &= \binom{N}{n}p^n (1-p)^{N-n} \\
	&= \binom{N}{\frac{1}{2}\left( \frac{d}{l} + N \right)} p^{\frac{1}{2}\left( \frac{d}{l} + N \right)}(1-p)^{N - \frac{1}{2}\left( \frac{d}{l} + N \right) } \\
	&= \binom{N}{\frac{d + N}{2l} } p^{\frac{d + Nl}{2l}}(1-p)^{\frac{Nl - d}{2l}}.
	\end{split}
	\end{equation*}
	Nevertheless this formula holds only for $d \in [-lN, lN]$. If $d > lN$ or $d<-Nl$, this probability must be zero because these regions of space cannot be reached by the random walk in $N$ time-steps. Restoring $Y_0$ (hence we simply translate the final position $d$ by $Y_0 = y_0$), we can write that
	\begin{equation}\label{RandomWalkProb|C}
	\begin{split}
	P[S_N - y_0 = d - y_0] = \binom{N}{\frac{d + N}{2l} } p^{\frac{d + N}{2l}}(1-p)^{\frac{N - d }{2l}},
	\end{split}
	\end{equation}
	Note that \eqref{RandomWalkProb|C} can be used as a probability only when the value of the random variable $Y_0$ is given: hence it is a conditional probability with respect to the value of $Y_0$, i.e. $P[S_N - y_0 = d - y_0] = P[S_N= d| Y_0 = y_0]$ . To complete the description of the random walk \eqref{RandomWalk}, using the Bayes theorem we obtain
	\begin{equation}\label{RandomWalkProb}
	P[S_N = d] = \sum_{y_0 \in \Int_l} P[S_N = d| Y_0 = y_0]\pi(y_0),
	\end{equation}
	which is the probability to find the random walk at time $N$ in the position $d \in \Int_l$, given that at the initial time it started from the position $Y_0$, random variable with distribution $\pi(y_0)$. Without loosing generality, we set $l =1$ for simplicity. We conclude our review on basic facts about the random walk, formalising the description at measure-theoretic level. As for any stochastic process, also for the random walk, there exists a probability space $(\Omega,\E,P)$. The sample space $\Omega$ can be imagined as the set of all possible trajectories of the random walk. It is a countable set (provided that the time of the random walk vary over a finite interval), since the random walk is a discrete process. $\E$ is a $\sigma$-algebra on $\Omega$, and can be thought as the power set of $\Omega$, i.e. $\E = \mathcal{P}(\Omega)$\footnote{ Given a set $A$, with the symbol $\mathcal{P}(A)$ we label the power set of $A$.}, while $P$ is the probability measure. The random walk on this probability space is the identity random variable evaluated at a given time $N$, i.e. for $s\in \Omega$ the position of the random walk at time $N$ is $S_N(s) = s(N)$.\newline
	
	Let us now come back to the space process. As stated in the beginning, it consists of a collection of $M$ random walks. At any time step $N$, the random distribution of points of the random walks is the space process of model A at time $N$. We may start with this preliminary definition.
	\begin{definition}
		Let $\{S^{(i)}_N\}_{i \in I} $ be a collection of independent random walks, where $|I| = M \in \Nat$, defined as in \eqref{RandomWalk} and described with the probability distributions \eqref{RandomWalkProb}. We call such a collection the \emph{space process} for the Model $A$.
	\end{definition}
	A possible realisation of the space process is given in figure \ref{fig:SP}.
	\begin{figure}[h!]
		\includegraphics[scale=0.5]{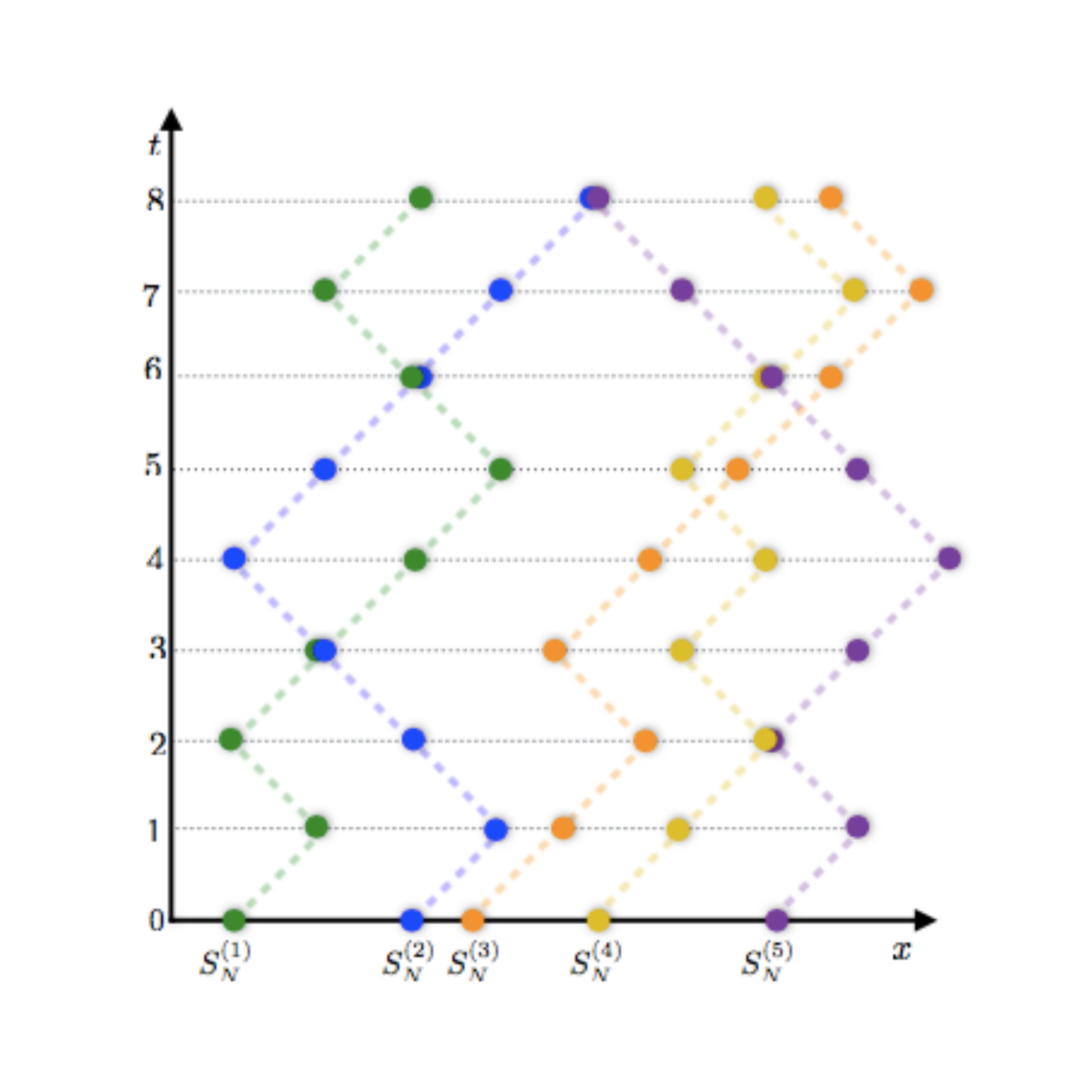}
		\caption{A realisation of an $M=5$ space process $\Spa$ is given. The set of coloured point at any time step represents the random distribution of points of the space process at a fixed time step, $\Spa_N$. The dashed lines liking the points of the same color represent the random walk evolution of each point of the space process. Such evolution changes the random distribution of points of the space process as time changes. Note that at a given time step, points may overlap.}
		\label{fig:SP}
	\end{figure}
	We label the space process of model A with the symbol $\Spa$, while $\Spa_N$ is the space process at the time-step $N$ (hence a random variable describing the distribution of $M$ points in $\Rea$). The outcome of the random variable $\Spa_N$, can be thought as a $M$-tuple, i.e. $\sss_N = (s_1(N),\cdots, s_M(N))$ where $s_i(N) \in \Rea$ is the position of the $i$-th random walk at time $N$. Call $\tilde{P}^A$ the probability measure for the space process $\Spa$. Since the random walks are assumed to be independent, the probability to obtain a specific configuration is given by
	\begin{equation}\label{Pro1}
	\tilde{P}^A[\Spa_N = \sss_N] = \prod_{i =1}^M P[S_N^{(i)} = s_{i}(N)].
	\end{equation}
	For the same reason, it may happen that for some realisation the points overlap. In a similar manner, we can also construct the joint probabilities
	\begin{equation}
	\begin{split}\label{Pro2}
	\tilde{P}^A[\Spa_N &= \sss_N, \Spa_T = \sss'_T] = \\
	&\prod_{i,j =1}^M P[S_N^{(i)} = s_{i}(N), S_T^{(j)} = s'_{j}(T)].
	\end{split}
	\end{equation}
	Note that $P[S_N^{(i)} = s_{i}(N), S_T^{(j)} = s'_{j}(T)]$ can be constructed using the independence of random walks when $i \neq j$, while for $i=j$ it is just the joint probability distribution of the $i$-th random walk. Proceeding in this way, we may construct the whole family of finite-dimensional distributions for the space process $\Spa$, which is consistent since the probabilities of the single random walks belongs to consistent families (in the sense of the Kolmogorov extension theorem, see Th. 2.1.5 in \cite{oksendal2013stochastic}). At this point we may replace the preliminary definition of the space process with the following which is more precise. 
	\begin{definition}
		Let $\{S_N^{(i)}\}_{i \in I}$ be  a collection of $M = |I| \in \Nat$ independent random walks defined on probability spaces $\{(\Omega_i,\E_i,P_i)\}_{i \in I}$. Let us define
		\begin{enumerate}
			\item[i)] $\Omega_{\Spa} :=\Omega_1 \times \cdots \times \Omega_M $;
			\item[ii)] $\E_{\Spa}= \mathcal{P}(\Omega_{\Spa})$;
			\item[iii)] $\tilde{P}^A: \E_{\Spa} \rightarrow [0,1]$ defined from the $\{P_i\}_{i \in I}$, as in \eqref{Pro1} or \eqref{Pro2} and generalisation.
		\end{enumerate}
		The space process is the stochastic process on $(\Omega_{\Spa},\E_{\Spa},\tilde{P}^A)$ defined as the identity function, namely $\Spa(s_1,\cdots,s_M ) = (s_1, \cdots, s_M)$.
	\end{definition}
	The set of all the possible configurations of points of the space process at a given time $N$ will be labeled by $\mathcal{S}(N)$.
	
	\subsection{The particle process}\label{sec3b}
	
	In this model, a particle is considered as a point-like object. At any time step $N$, it is completely described by its \emph{position} and its \emph{velocity}, which are assumed to be random variables.\newline
	
	The \emph{position random variable}, labeled by $X_N$, is interpreted as the actual position of the particle at time $N$. Let $(\Omega_{\Spa},\E_{\Spa},\tilde{P}^A)$ be the probability space for the space process. On a probability space $(\Omega_I,\E_I,P_I)$ define an integer value discrete-time stochastic process $I_N: \Omega_I \rightarrow \{1,\cdots, M\}$, which we call \emph{selection process}. Assume that we place the origin of a reference frame in $\Spa_N$ in the point $S^{(i_O)}_N(\sss)$. Then we define
	\begin{equation}\label{PosProc}
	X_N(\omega) := \pi_{I_N(\omega_I)}(\Spa_N(\sss)) - S^{(i_O)}_N(\sss),
	\end{equation}
	where $\pi_i$ is the projector of the $i$-th component of an $M$-tuple, and $\omega = (\omega_I, \sss)$ with $\omega_I \in \Omega_I$ and $\sss \in \Omega_{\Spa}$. Hence we can say that $X_N = S^{(i_N)}_N - S^{(i_O)}_N$ for some $i_N \in \{1,\cdots,M\}$ and $S_N^{(i)},S^{(i_O)}_N \in \sss_N$ components of a given realisation of the space process at time $N$ (the writing \textquotedblleft $x \in \mathbf{y}$", where $x$ is a point and $\mathbf{y}$ is an $N$-tuple, should be interpreted as \textquotedblleft $x$ is a member of the $N$-tuple $\mathbf{y}$"). Thus we have the following definition:
	\begin{definition}
		Consider the probability space $(\Omega_I \times \Omega_{\Spa}, \E_I \otimes \E_{\Spa}, P^A)$ and a measurable space $(\Int, \mathcal{P}(\Int))$. The random variable $X_N$ is the $\mathcal{P}(\Int)$-measurable function
		\begin{equation*}
		X_N: \Omega_I \times \Omega_{\Spa} \rightarrow \Int
		\end{equation*} 
		defined as in \eqref{PosProc}. $X_N$ represents the position of the particle at time $N$.
	\end{definition}
	Note that on the probability space $(\Omega_I \times \Omega_{\Spa}, \E_I \otimes \E_{\Spa}, P^A)$ we can describe also the space process $\Spa$ by simply demanding that $P^A \circ [\Spa]^{-1} = \tilde{P}^A$. From now on in the whole discussion of model A, instead of writing $P^A$ we simply write $P$ if no confusion arises. By construction, $X_N$ is a function of the space process $\Spa$. This implies that $X_N$ and $\Spa_N$ are not two independent random variables. Indeed, assume
	\begin{equation*}
	\begin{split}
	\Spa_N(\sss) &= \sss_N := (x_1, \cdots, x_M), \\
	\Spa_N(\sss') &= \sss_N' := (x_1',\cdots,x_M'),
	\end{split}
	\end{equation*}
	and that we can fix a common origin on them, say $x_o \in \sss_N \cap \sss_N'$. Choose $\sss$ and $\sss'$ such that there exists $z + x_o \in \sss_N$ but $z + x_o  \notin \sss_N'$, i.e. $\Spa_N(\sss)$ and $\Spa_N(\sss')$ have at least one point which is not in common. Then
	\begin{equation*}
	P[X_N = z | \Spa_N = \sss_N] \neq P[X_N = z | \Spa_N =\sss_N'],
	\end{equation*}
	since the second term vanishes by construction while the first can be non-zero in general. Thus we cannot set $P[X_N = z | \Spa_N = \sss_N] = P[X_N = z]$ in general, which implies
	\begin{equation*}
	P[X_N = z, \Spa_N = \sss_N] \neq P[X_N = z]P[\Spa_N = \sss_N]
	\end{equation*}
	
	Let us now describe the \emph{velocity random variable}. In order to introduce this process, we need to specify how the particle moves on a physical space described with the space process introduced before. We assume that particle moves by jumps: it jumps from one of the points of the space process at time $N$ to another point of the space process at time $N+1$. These jumps are described by the transition probabilities
	\begin{equation}\label{TransProb}
	P[X_{N+1} = b | X_N = a] = \alpha(b,a),
	\end{equation}
	where $a,b \in \Int$. Once these transition probabilities are given, we can define the velocity random variable $V_N$. We set
	\begin{equation}\label{VeloProc}
	V_N := \frac{X_{N+1} - X_N}{N+1 -N} = X_{N+1} - X_N.
	\end{equation}
	This is clearly the discrete-time version of the usual definition of velocity. Note that this physical definition makes sense because, thanks to the transition probabilities \eqref{TransProb}, we can describe $V_N$ from the probabilistic point of view using only information available at time $N$. More formally, the transition probabilities \eqref{TransProb} allows to describe $V_N$ on the same probability space of $X_N$, i.e. $(\Omega_I \times \Omega_{\Spa}, \E_I \otimes \E_{\Spa}, P^A)$.
	\begin{definition}
		Consider the probability space $(\Omega_{I}\times \Omega_{\Spa}, \E_I \otimes \E_{\Spa}, P^A)$ and the measurable space $(\Int, \mathcal{P}(\Int))$. The velocity random variable $V_N$ is the $\mathcal{P}(\Int)$-measurable function
		\begin{equation*}
		V_N: \Omega_{I}\times \Omega_{\Spa} \rightarrow \Int
		\end{equation*}
		defined in \eqref{VeloProc}. $V_N$ represents the velocity of the particle at time $N$.
	\end{definition}
	Also the velocity random variable is a function of the space process and, proceding as done for $X_N$, we may conclude that $P[V_N =c , \Spa_N = \sss_N] \neq P[V_N =c]P[\Spa_N = \sss_N]$. Let us now derive the relation between the probabilities $P[V_N = c]$ and $P[X_N = a]$, in a way that is consistent with the transition probabilities \eqref{TransProb}. It can be done following this intuitive idea. Suppose that at time $N$ we know that the particle is in the position $X_N = a$. Then the event $A:= \{X_N = a\}$ is true, i.e. $P(A) = P[X_N = a] = 1$, which means that $P[X_N = a'|A] = \delta_{a,a'}$. Under the same conditions, one should also write that $V_N = X_{N+1} - a$, and this suggests that the probability to observe $V_N = c$ is equal to the probability to observe $X_{N+1} = a+c$, when $A$ happens. Thus, using \eqref{TransProb} we can write that
	\begin{equation*}
	P[V_N=c|A] = P[X_{N+1}=a+c|A] = \alpha(a+c,a).
	\end{equation*}
	The equation above can be confirmed in a more rigorous way.
	\begin{proposition}\label{conditionalVel}
		Let $X_N$  and $V_N$ be the position and the velocity random variables. If $P[X_{N+1} = b|X_N =a] = \alpha(b,a)$, then $P[V_N = c| A] = \alpha(a+c,a)$ where $A = \{X_N =a\}$.
	\end{proposition}
	\begin{proof}
		Since $V_N = X_{N+1} - X_N$, clearly $X_{N+1}$ and $X_{N}$ are conditionally independent under the event $A = \{X_N = a\}$. Let $\varphi_{V_N}(\lambda)|_A$, $\varphi_{X_{N+1}}(\lambda)|_A$ and $\varphi_{X_N}(\lambda)|_A$ be the characteristic functions of the three random variables considered here, computed with the conditional probabilities. By conditional independence we can write that
		\begin{equation*}
		\varphi_{V_N}(\lambda)|_A = \varphi_{X_{N+1}}(\lambda)|_A \cdot\varphi_{-X_N}(\lambda)|_A.
		\end{equation*}
		Since
		\begin{equation*}
		\begin{split}
		\varphi_{-X_N}(\lambda)|_A &= \sum_{a'} e^{-i \lambda a'}\delta_{a,a'} = e^{-i\lambda a} \\
		\varphi_{X_{N+1}}(\lambda)|_A &= \sum_{b} \alpha(b,a)e^{i \lambda b}
		\end{split}
		\end{equation*}
		we have that
		\begin{equation*}
		\varphi_{V_N}(\lambda)|_A = \sum_{b} \alpha(b,a)e^{i\lambda (b - a)}.
		\end{equation*}
		Because $b-a \in \Int$, clearly $\varphi_{V_N}(\lambda)|_A = \varphi_{V_N}(\lambda + 2\pi)|_A$ which means that the random variable $V_N$ is a discrete random variable (as expected). The inversion formula of the characteristic function, in this case is
		\begin{equation*}
		P[V_N = c| A] = \lim_{T \rightarrow +\infty} \frac{1}{2T} \int_{-T}^{+T} e^{-i \lambda c}\varphi_{V_N}(\lambda)|_A d\lambda.
		\end{equation*}
		Thus
		\begin{equation*}
		\begin{split}
		P[V_N = c| A] &= \lim_{T \rightarrow +\infty} \frac{1}{2T} \int_{-T}^{+T} e^{-i \lambda c}\sum_{b} \alpha(b,a)e^{i\lambda (b - a)} d\lambda \\
		&= \sum_{b} \alpha(b,a) \lim_{T \rightarrow +\infty} \frac{1}{2T} \int_{-T}^{+T} e^{i\lambda (b - a -c)} d\lambda \\
		&= \sum_{b} \alpha(b,a) \lim_{T \rightarrow +\infty} \frac{e^{iT (b - a -c)} - e^{iT (b - a -c)}}{2Ti(b-a-c)} \\
		&= \sum_{b} \alpha(b,a) \lim_{T \rightarrow +\infty} \mbox{sinc}(T(b-a-c))
		\end{split}
		\end{equation*}
		where $\mbox{sinc}(x) = \sin x / x$. Since $\lim_{a \rightarrow \infty} \mbox{sinc}(ax) = \delta_{x,0}$ when $x \in \Int$, we conclude that
		\begin{equation*}
		P[V_N = c| A] = \sum_{b} \alpha(b,a) \delta_{b-a-c , 0} = \alpha(a+c,a).
		\end{equation*}
		This concludes the proof.
	\end{proof}
	At this point, we may obtain $P[V_N = c]$ simply using the Bayes theorem, namely
	\begin{equation}\label{BayesV}
	P[V_N = c] = \sum_a \alpha(a+c,a)P[X_N =a]
	\end{equation}
	which is consistent with the transition probabilities given in the beginning. The following assumption on the transition probabilities is done
	\begin{equation}\label{transiprobi}
	P[V_N = c | X_N =a] = P[X_N = a | V_N = c],
	\end{equation}
	i.e. the transition probabilities are \emph{symmetric} under the exchange of their arguments. Having defined both the position and velocity random variables, we may give a precise definition of what we call particle in Model A .
	\begin{definition}
		A particle is a point like-object whose features at time $N$ are completely specified by the position and velocity random variables. More formally, we can say that a particle corresponds to the random vector $\mathscr{P}_N := (X_N,V_N)$. We will refer to $\mathscr{P}_N$ with the name \emph{particle process}, when considered as a function of time.
	\end{definition}
	An example of particle process is drawn in figure \ref{fig:PP}.
	\begin{figure}[h!]
		\includegraphics[scale=0.50]{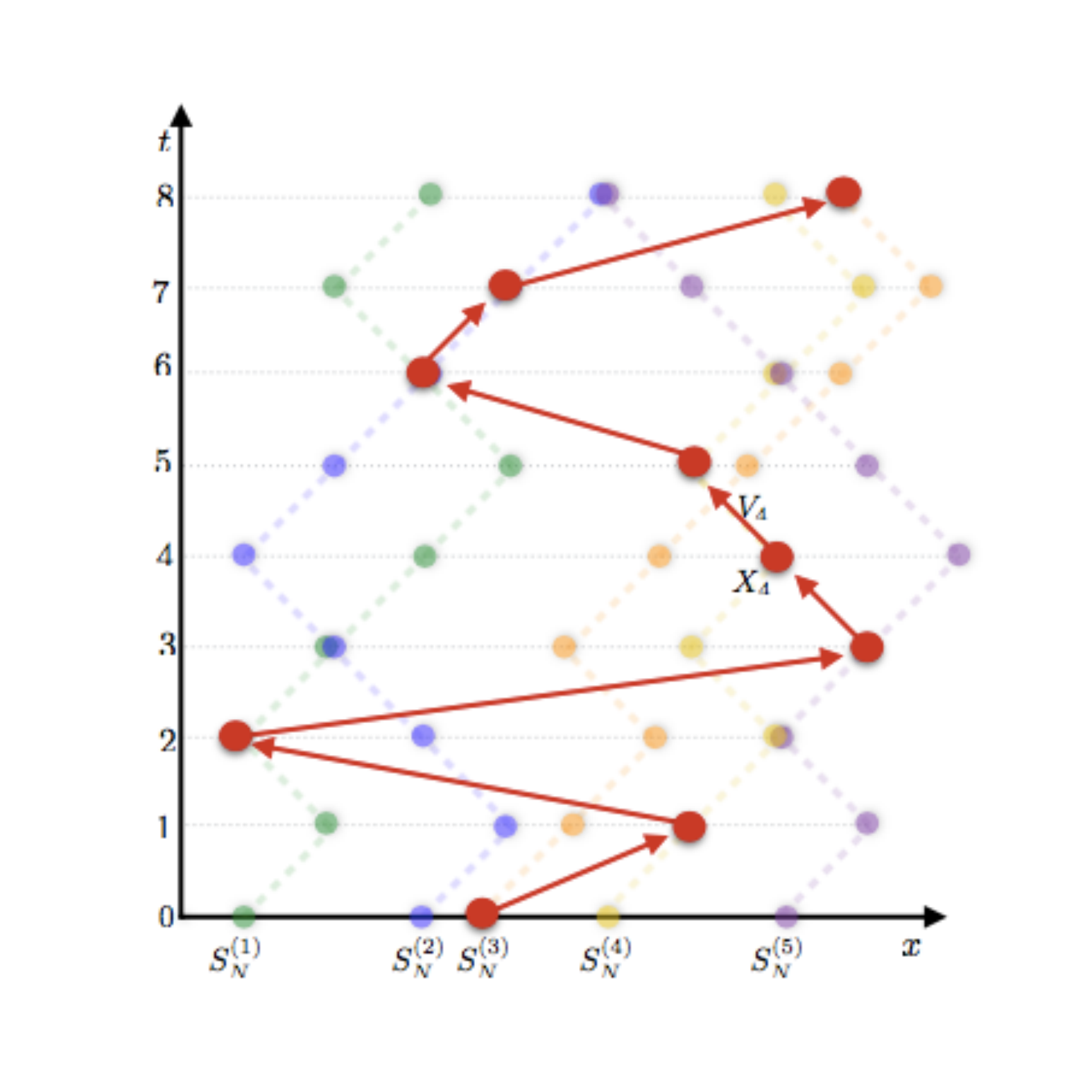}
		\caption{A possible realization of the particle process is drawn in red, over the same realization of the space process considered in figure \ref{fig:SP}. The position of the particle at a given time-step is given by the red point while its velocity at the same time-step is represented by the outgoing arrow.}
		\label{fig:PP}
	\end{figure}

	\subsection{Remove the space from the model}
	
	In this section, we will explain what we mean with the expression \textquotedblleft remove the space from the model". In \cite{LC} we observed that, given three random variables $X$, $Y$ and $Z$ on the same probability space $\PS$, one can eliminate one of them, say $Z$, simply by conditioning with respect to the outcomes of this random variable, i.e. conditioning on the events $\{Z=z\}$. In the collection of probability spaces obtained after conditioning, the description of the random variable $Z$ is not anymore possible unless one adds the probabilities $P[Z=z]$. Such information cannot be obtained from the collection of probability spaces one has after conditioning. The statistical description of the remaining random variables can be done without influence the random variable $Z$: in this sense $Z$ is not present anymore in the probabilistic model used to describe $X$ and $Y$. However, when we deal with a stochastic process the elimination of a random variable representing it at given time, do not guarantee that we can manipulate all the remaining random variables without influence the stochastic process (which means that we can describe the remaining random variables without influence the removed process). In this case we need to add additional conditions in order to be sure that the stochastic process is not present anymore in the remaining probabilistic model. Once that we apply a procedure that is capable to do so, we say that  the stochastic process is \emph{removed from the model}.
	
	Model A exhibits features that are interesting from the point of view of quantum mechanics when we remove the space process from the model. Before describing how to implement it mathematically, let us first explain the physical principles that motivate this removal. Model A describes a particle that jumps at random over a random distribution of points. Such a random distribution of points is assumed to be the physical space in which the particle moves: the physical space is not anymore a passive background against which physical processes take place. Preparing the particle in a given state means to perform an experimental procedure after which the statistical properties of the particle's observables are known. In other words, the state preparation is an experimental procedure such that right after it terminates, all random variables associated to the particle's observables have a given probability distribution. Hence, saying that a particle at time $N$ is in a given state, means that at time $N$ all probability distributions of the observables of the particle are fixed. However, assuming that the physical space is random has big consequences. Any experimental procedure happens in such a physical space. If the experimental procedure for the state preparation ends at time $N$, the probability distributions of the observables are always conditioned to the configuration of space at that time. This is because one prepares the state of the particle at time $N$, in the configuration of space that the space process assumes at that time. Hence the probability distributions that describe the particle must be always conditioned to some space configuration. If it is not so, to prepare the particle in a given state we need to have control not only on it but also on the whole space. This means that the probability distribution that describes the space process at a given time does not depend on the probabilities describing the particle after a preparation procedure. In other words, changing the probability distributions describing the particle (i.e. changing the state) does not have to modify the probability distributions describing the space process. When this happens we say that the space process is removed from model A. In order to implement that, we have to require the following:
	\begin{enumerate}
		\item[i)] The particle at time $N$ can be described only by using probabilities that are conditioned with respect to some space configuration at that time. This means that to describe the position and velocity random variables we have to use only
		\begin{equation*}
		\begin{split}
		P_{\sss_N}[X_N = a] &:= P[X_N = a | \Spa_N = \sss_N], \\
		P_{\sss_N}[V_N = c] &:= P[V_N = c | \Spa_N = \sss_N],
		\end{split}
		\end{equation*}
		where $\sss_N$ is the configuration of the space process at time $N$.
		\item[ii)] The transition probabilities of any point of space (i.e. $p_{i} = P[S^{(i)}_{N+1} = a + 1|S^{(i)}_N = a]$ for all $i \in I$) cannot be changed by the preparation procedure of the particle. This means that changing the conditional probabilities of the particles, the transition probabilities of the single point of space remains fixed.
	\end{enumerate}
	These two conditions implement the idea that the space process cannot be influenced by the preparation procedure of the particle. Note that the requirement $i)$ is needed in order to avoid that the probability of the space process at time $N$ is changed by the preparation procedure of the particle, while the requirement $ii)$ avoids that such preparation procedure alters the space process probabilities at times $N' \neq N$ (i.e. in the past or in the future). Since we are dealing with non-relativistic systems this last requirement is reasonable from the physical point of view.
	
	Let us now describe the effects of the removal of the space process in model A from the mathematical point of view. We will focus first on the consequences of the requirement $i)$. In order to do so we need to study better the effect of conditioning on a probability space. For the interested reader, in appendix C a short review on how conditioning is described in the measure-theoretic formulation of probability theory is presented. However here we proceed following a more intuitive approach. According to the removal procedure explained above, we can describe the particle using only probabilities that are conditioned to the event $\{\sss_N\}:=\{ \sss \in \Omega_{\Spa} | \Spa_N(\sss) = \sss_N \}$. At the level of the events, this means that for the random variable $X_N$ and $V_N$ we consider only events of this kind: $\{X_N \in A\} \cap \{ \Spa_N = \sss_N \}$ and $\{V_N \in B\} \cap \{ \Spa_N = \sss_N \}$. For the position random variable, this means that the conditioning procedure effectively changes the sample space and the $\sigma$-algebra of its starting probability space as
	\begin{equation*}
	(\Omega_I \times \Omega_{\Spa}, \E_I \otimes \E_{\Spa}) \rightarrow ( \Omega_I \times \{\sss_N\}, \mathcal{P}(\Omega_I \times \{\sss_N\}) ).
	\end{equation*}
	Let us call $\Omega_{X_N}:=\Omega_I \times \{\sss_N\}$ and $\E_{X_N}:=\mathcal{P}(\Omega_I \times \{\sss_N\})$. It is a known fact from probability theory that the measurable space $(\Omega_{X_N}, \E_{X_N})$ equipped with conditional probability $P_{\sss_N}[X_N = \cdot]$ defines a probability space. On this probability space $(\Omega_{X_N},\E_{X_N},P_{\sss_N})$ the random variable $X_N$ can be described after conditioning on the event $\{\sss_N\}$. Everythig we said till now, clearly also holds for the velocity random variable: after conditioning it can be described in a probability space $(\Omega_{V_N},\E_{V_N},P_{\sss_N})$ defined in a similar manner. 
	
	Relevant for our goal is the study of the joint probabilities for $X_N$ and $V_N$, and its link with the transition probabilities \eqref{TransProb} after conditioning. By definition $X_N$ and $V_N$ are two random variables defined on the same probability space $(\Omega_{I}\times \Omega_{\Spa}, \E_I \otimes \E_{\Spa}, P)$. This means that we can always find a joint probability distribution $P[X_N=a,V_N=c]$, which can be used to derive the transition probabilities \eqref{TransProb} using the usual Bayes formula. Since the space process can be described on the same probability space of $X_N$ and $V_N$, also the joint probability distribution $P[X_N = a, V_N = c, \Spa_N = \sss_N]$ exists. Applying the Bayes formula, we can derive the \emph{conditional} joint probability for $X_N$ and $V_N$, namely
	\begin{equation*}
	P_{\sss_N}[X_N=a,V_N=c] := \frac{P[X_N = a, V_N = c, \Spa_N = \sss_N]}{P[\Spa_N = \sss_N]},
	\end{equation*}
	from which one can derive \emph{conditional} transition probabilities
	\begin{equation*}
	\begin{split}
	\alpha_{\sss_N}(c,a) &:= P_{\sss_N}[V_N = c|X_N =a] \\
	&= \frac{P_{\sss_N}[X_N=a,V_N=c]}{P_{\sss_N}[X_N = a]}.
	\end{split}
	\end{equation*}
	Note that $\alpha_{\sss_N}(c,a) \neq \alpha(c,a)$. From the point of view of the probability spaces, after conditioning we can always describe the two random variables using a single probability space. Such probability space is simply $(\Omega_{X_N} \times \Omega_{V_N}, \E_{X_N} \otimes \E_{V_N}, P_{\sss_N})$. On it, we can define a joint probability distribution $P_{\sss_N}[X_N = \cdot , V_N = \cdot ]$ such that $P_{\sss_N}[X_N = \cdot]$ and $P_{\sss_N}[V_N = \cdot]$ are the two marginals and $\alpha_{\sss_N}(c,a)$ are the transition probabilities between $V_N$ and $X_N$. Since the joint probability distribution are symmetric under the exchange of the arguments, clearly
	\begin{equation*}
	\alpha_{\sss_N}(c,a) P_{\sss_N}[X_N =a] = \alpha_{\sss_N}(a,c) P_{\sss_N}[V_N = c],
	\end{equation*}
	where $\alpha_{\sss_N}(a,c) = P_{\sss_N}[X_N = a|V_N =c]$. According to \cite{khrennikov2009contextual}, this is a signature that we are working on a single measure-theoretic probability space. A more interesting case happens when we use the the \emph{unconditional} transition probabilities $\alpha(a+c,a)$ and $P[X_N = a | V_N =c]$. In this case we have that
	\begin{equation}\label{noBayes}
	\alpha(a+c,a)P_{\sss_N}[X_N =a] \neq P[X_N = a | V_N =c] P_{\sss_N}[V_N = c]
	\end{equation}
	in general, which means that we cannot describe $X_N$ and $V_N$ using a single measure-theoretic probability space, \emph{ if we choose to use the unconditional transition probabilities after conditioning with respect to the space process at time $N$}. However, this does not mean that we cannot describe $X_N$ and $V_N$ after conditioning using the transition probabilities $\alpha(a+c,a)$ (we will come back on the physical reason for the use of $\alpha(a+c,a)$ instead of  $\alpha_{\sss_N}(c,a)$ later). We can do it using two different probability spaces: one for $X_N$ and one for $V_N$. We have already seen that, after conditioning, we obtain a probability spaces for each random variables, i.e. $(\Omega_{X_N},\E_{X_N},P_{\sss_N})$ for $X_N$ and $(\Omega_{V_N},\E_{V_N},P_{\sss_N})$ for $V_N$. However we cannot construct a joint probability space where $P_{\sss_N}[X_N = \cdot]$ and $P_{\sss_N}[V_N = \cdot]$ are the two marginals of some joint probability distribution and $\alpha(a+c,a)$ are the transition probabilities that we obtain from the same joint probability distribution. This is exactly the content of \eqref{noBayes}: the joint probability we are looking for would not be symmetric in the exchange of the arguments. This is something that it is not possible in an ordinary measure space since the intersection of events in a sigma algebra is a symmetric operation (i.e. commutative). As a consequence we may conclude that the Bayes theorem cannot be used to relate the two marginals. However a relation between $P_{\sss_N}[X_N = \cdot]$ and $P_{\sss_N}[V_N = \cdot]$ can still be found \cite{khrennikov2005interference}.
	\begin{theorem}\label{theoV1}
		Let $\{P_{\sss_N}[X_N = a]\}_{a \in \Omega_{X_N}}$ be the probabilities describing the position of the particle at time $N$ under the condition that the space process at time $N$ is $\sss_N$. If $P[X_{N+1}=b|X_N =a] = \alpha(b,a)$, then
		\begin{equation}\label{VelocityProb}
		P_{\sss_N}[V_N =c] = \sum_a \alpha(a+c,a)P_{\sss_N}[X_N = a] + \delta(c|X_N,\sss_N)
		\end{equation}
		where
		\begin{equation}\label{Velocity compl}
		\begin{split}
		\delta(c| X_N, \sss_N) &= \frac{1}{P[\Spa_N = \sss_N]} \sum_{\substack{\sss_N'\\ \sss_N' \neq \sss_N}} \bigg[ \sum_a \alpha(a+c,a)\cdot \\
		& \cdot P[X_N = a, \Spa_N = \sss_N'] - P[V_N=c, \Spa_N = \sss_N']\bigg]
		\end{split}
		\end{equation}
		which is in general different from zero.
	\end{theorem}
	\begin{proof}
		Given $P[V_N =c]$, we can always write
		\begin{equation*}
		\begin{split}
		P[V_N = c] = \sum_{\sss_N'} P[V_N = c, \Spa_N = \sss_N']
		\end{split}
		\end{equation*}
		and similarly
		\begin{equation*}
		P[X_N =a] = \sum_{\sss_N'} P[X_N = a, \Spa_N = \sss_N'].
		\end{equation*}
		Note that the sum over all possible configurations of the space process at time $N$ is well defined, since the number of configurations is clearly countable (it is a cartesian product of a discrete process taking value on the integers). Substituting these expressions in \eqref{BayesV} and dividing by $P[\Spa_N = \sss_N]$, we get
		\begin{widetext}
			\begin{equation*}
			\begin{split}
			\sum_{\sss_N'} \frac{P[V_N = c, \Spa_N = \sss_N']}{P[\Spa_N = \sss_N]} &= \sum_a \alpha(a+c,c)\sum_{\sss_N'} \frac{P[X_N = a, \Spa_N = \sss_N']}{P[\Spa_N = \sss_N]} \\
			P_{\sss_N}[V_N =c] + \sum_{\substack{\sss_N'\\ \sss_N' \neq \sss_N}} \frac{P[V_N = c, \Spa_N = \sss_N']}{P[\Spa_N = \sss_N]} &= \sum_a \alpha(a+c,c)P_{\sss_N}[X_N = a] + \sum_a \alpha(a+c,c)\sum_{\substack{\sss_N' \\\sss_N' \neq \sss_N}} \frac{P[X_N = a, \Spa_N = \sss_N']}{P[\Spa_N = \sss_N]}
			\end{split}
			\end{equation*}
		\end{widetext}
		Moving the second term of the LHS to the RHS, we obtain \eqref{VelocityProb} and \eqref{Velocity compl}. Note that in general \eqref{Velocity compl} is non zero since
		\begin{equation*}
		\begin{split}
		\sum_a \alpha&(a+c,a)P[X_N = a, \Spa_N = \sss_N'] \\
		&\neq P[V_N=c, \Spa_N = \sss_N'].
		\end{split}
		\end{equation*}
		This concludes the proof.
	\end{proof}
	We can see that, after the conditioning on the space process, the Bayes formula cannot be used anymore to compute $P_{\sss_N}[V_N = c]$ from the probabilities of the position random variable if we want to use the transition probabilities $\alpha(a+c,a)$. We need to add a correction term which contains statistical information about the space process. Note that this correction term has the property
	\begin{equation}\label{propDelta}
	\sum_c \delta (c | X_N, \sss_N) = 0,
	\end{equation}
	which is necessary in order to preserve the normalisation of probabilities, i.e. $\sum_cP_{\sss_N}[V_N = c] = 1$. We also note that in general $\delta(c | X_N, \sss_N) \in [-1,1]$ and in particular it can be negative. Summarising, given the transition probabilities $\alpha(a+c,a)$ we cannot describe $X_N$ and $V_N$ on a single probability space after conditioning on the space configuration at time $N$. However, the description $X_N$ and $V_N$ in a single probability space after conditioning can be always done: the price to pay is that we have to change the transition probabilities from $\alpha(a+c,a)$ to $\alpha_{\sss_N}(c,a)$.
	
	At this point a legitimate question arises: can we motivate physically the choice to use $\alpha(a+c,a)$ instead $\alpha_{\sss_n}(c,a)$? Yes, if we take into account the fact that we want to remove space from the model. Indeed, in order to measure with an experimental procedure $\alpha_{\sss_N}(c,a)$, one would have control over space since one has to be able to prepare the space process always in the configuration $\sss_N$, in order to measure $\alpha_{\sss_N}(c,a)$. Since the removal of space is done exactly to avoid such things, the use of $\alpha(a+c,a)$ is more reasonable from the physical point of view. We want to conclude our analysis on the consequence of the requirement $i)$ with a comment on the particle process. Since it is a random vector parametrized by time, one may be tempted to consider $\mathscr{P}_N$ as a stochastic process. This is certainly possible considering also the space process, namely before conditioning on $\sss_N$. Nevertheless, after conditioning and using the transition probabilities $\alpha(a+c,a)$, we just have a collection of probability spaces and it is not trivial to assume that each of these spaces can be seen as, part of a bigger probability space describing the particle only (i.e. with no space process involved in the construction of such probability space) as the Kolmogorov extension theorem \cite{oksendal2013stochastic} would imply. For this reason, considering the particle process as a stochastic process, in this context, should be done with care.\newline
	
	Till now we explored the consequences of the requirement $i)$ for the removal of space. Conditioning with respect to the space configuration $\sss_N$, we effectively eliminate the possibility to change the $P[\Spa_N=\sss_N]$ by varying the (conditional) probability distributions in the collection of probability spaces that describe the model. The requirement $ii)$ is added in order to avoid that by varying the probabilities of the particle we can modify the probabilities $P[\Spa_{N'} = \sss_{N'}]$ when $N' \neq N$. The consequences of $ii)$ which are relevant for our analysis will be analyzed in the next section.

	\subsection{The entropic uncertainty relation for $X_N$ and $V_N$}
	
	In this section we will analyze the basic consequence of the requirement $ii)$ for the removal of the space process in model A. From now on, we exclude that the probabilities describing the space process and its constituents have delta-like distributions. This implies that the space process of model A is not a deterministic process. We note that the whole removal procedure, which makes model A interesting to study, is meaningless in this case. The central result of this section is the following. 
	
	\begin{theorem}\label{EURmodelA}
		Let $X_N$ and $V_N$ be the position and velocity random variables of model A. Fixing the transition probabilities $p_i= P[S_{N+1}^{(i)} = a+1|S_N^{(i)} = a]$ of the points of the space process for all $i \in I$, then
		\begin{equation}\label{EUR-ModA}
		H_{\sss_N}(X_N) + H_{\sss_N}(V_N) \geqslant D,
		\end{equation}
		where $D$ is a positive constant which does not depend on $\{P_{\sss_N}[X_N = a]\}_{a \in X_N(\Omega_X)}$ and $\{P_{\sss_N}[V_N = c]\}_{c \in V_N(\Omega_V)}$.
	\end{theorem}
	\begin{proof}
		The entropy is a non-negative quantity by definition, hence varying with respect to all $P_{\sss_N}[X_N = a]$ clearly $H_{\sss_N}(X_N) \geqslant 0$. Now consider the entropy for the random variable $V_N$ and let us study what happens when we vary with respect to all $P_{\sss_N}[X_N=a]$. Given $P_{\sss_N}[X_N = a]$, the probability $P_{V_N}[V_N=c]$ can be computed by means of the formula in theorem \ref{theoV1}. On the other hand we are always free to use the conditional transition probabilities $\alpha_{\sss_N}(c,a)$, i.e. to work on the joint probability space of $X_N$ and $V_N$, to study how $H_{\sss_N}(V_N)$ change varying with respect to $P_{\sss_N}[X_N =a]$.This allow us to write
		\begin{equation*}
		H_{\sss_N}(V_N) \geqslant \sum_a P_{\sss_N}[X_N = a]H_{\sss_N}(V_N|X_N=a),
		\end{equation*}
		with
		\begin{equation*}
		H_{\sss_N}(V_N|X_N=a) = - \sum_c \alpha_{\sss_N}(c,a) \log \alpha_{\sss_N}(c,a).
		\end{equation*}
		The conditional transition probabilities $\alpha_{\sss_N}(c,a)$ can be rewritten as follows. Consider the joint probability $P_{\sss_N}[X_{N+1} = b, X_N =a]$. In what follows, without loss of generality we set $S_N^{i_O} = 0$ at any time $N$. We can write the following
		\begin{widetext}
			\begin{equation*}
			\begin{split}
			&\mspace{80mu} P_{\sss_N}[X_{N+1} = b, X_N = a] = \sum_{i = 1}^M P_{\sss_N}[X_{N+1} = S_{N+1}^{(i)}| S_{N+1}^{(i)} = b, X_N =a] P_{\sss_N}[S_{N+1}^{(i)} = b, X_N =a] \\
			&\mspace{50mu}=\sum_{i = 1}^M P_{\sss_N}[X_{N+1} = S_{N+1}^{(i)}| S_{N+1}^{(i)} = b, X_N =a] \bigg(\sum_{j=1}^M P_{\sss_N}[S_{N+1}^{(i)} =b| X_N =a,I_N = j]P_{\sss_N}[X_N =a,I_N = j]\bigg) \\
			&=\sum_{i = 1}^M P_{\sss_N}[X_{N+1} = S_{N+1}^{(i)}| S_{N+1}^{(i)} = b, X_N =a] \bigg(\sum_{j=1}^M P_{\sss_N}[S_{N+1}^{(i)} = b| X_N =a,I_N = j]P_{\sss_N}[I_N = j|X_N =a]   P_{\sss_N}[X_N = a]\bigg)
			\end{split}
			\end{equation*}
			Since the event $\{X_N =a\} \cap \{I_N = j\} = \{S_N^{(j)} = a\}$ by definition and because $P_{\sss_N}[X_{N+1} = a+c,X_N =a] = P_{\sss_N}[V_N = c, X_N =a]$, from this decomposition we can conclude that
			\begin{equation*}
			\alpha_{\sss_N}(c,a) =\sum_{i = 1}^M P_{\sss_N}[X_{N+1} = S_{N+1}^{(i)}| S_{N+1}^{(i)} = b, X_N =a] \bigg(\sum_{j=1}^M P_{\sss_N}[S_{N+1}^{(i)} = b| S^{(j)}_N =a]P_{\sss_N}[I_N = j|X_N =a] \bigg).
			\end{equation*}
			We also note that
			\begin{equation}\label{lemma1}
			\begin{split}
			\sum_{i = 1}^M P_{\sss_N} &[X_{N+1} = S_{N+1}^{(i)}| S_{N+1}^{(i)} = b, X_N =a]  = 1, \\
			&\sum_{j=1}^M P_{\sss_N}[I_N = j|X_N =a] = 1.
			\end{split}
			\end{equation}
			In what follows, we set $\gamma(i):= P_{\sss_N}[X_{N+1} = S_{N+1}^{(i)}| S_{N+1}^{(i)} = b, X_N =a]$ and $\eta(i,j):= P_{\sss_N}[S_{N+1}^{(i)} = b| S^{(j)}_N =a]$ in order to keep the notation compact. From the above decomposition of $\alpha_{\sss_N}(c,a)$ we can write that
			\begin{equation*}
			\begin{split}
			H_{\sss_N}& (V_N|X_N = a) =  - \sum_c \left( \sum_{i=1}^M \gamma(i)\bigg(\sum_{j=1}^M \eta(i,j)P_{\sss_N}[I_N = j|X_N =a] \bigg) \right)
			\log \left( \sum_{i=1}^M \gamma(i)\bigg(\sum_{j=1}^M \eta(i,j)P_{\sss_N}[I_N = j|X_N =a] \bigg)\right).
			\end{split}
			\end{equation*}
			Note that since only positive probabilities contribute to the entropy, all the $\alpha_{\sss_N}(c,a)$ are different from zero. This implies that all the $\gamma(i)$, $\eta(i,j)$ and $P_{\sss_N}[I_N = j|X_N =a]$ used to compute the entropy are strictly positive. Since $f(x) = -x\log x$ is a concave function, by the Jensen inequality and using \eqref{lemma1}, we have
			\begin{equation*}
			\begin{split}
			H_{\sss_N}(V_N|X_N = a) &\geqslant \sum_c \sum_{i=1}^M \gamma(i)\left( -\bigg(\sum_{j=1}^M\eta(i,j)P_{\sss_N}[I_N = j|X_N =a] \bigg)
			\log \bigg(\sum_{j=1}^M \eta(i,j)P_{\sss_N}[I_N = j|X_N =a] \bigg) \right) \\
			&\geqslant   \sum_c \left( \sum_{i=1}^M \gamma(i) \right) \min_i\left(  - \bigg(\sum_{j=1}^M \eta(i,j)P_{\sss_N}[I_N = j|X_N =a] \bigg) 
			\log \bigg(\sum_{j=1}^M\eta(i,j)P_{\sss_N}[I_N = j|X_N =a] \bigg) \right) \\
			&\geqslant \sum_c \min_i \left( \bigg(\sum_{j=1}^M P_{\sss_N}[I_N = j|X_N =a]\bigg) \min_j \left( - \eta(i,j) \log \eta(i,j)\right) \right) \\
			& = \sum_c \min_{i,j} \left( - \eta(i,j) \log \eta(i,j)\right)
			\end{split}
			\end{equation*}
			where $\min_{i,j}$ means the minimum over $i,j \in \{1, \cdots, M\}$ keeping $c$ constant. Summarising, we have that
			\begin{equation*}
			\begin{split}
			H(V_N|X_N = a) \geqslant  \sum_c \min_{i,j} \left( - P_{\sss_N}[S_{N+1}^{(i)} = a+c| S^{(j)}_N =a]\log P_{\sss_N}[S_{N+1}^{(i)} = a+c| S^{(j)}_N =a] \right).
			\end{split}
			\end{equation*}
			Note that in the RHS there is still a dependence on $a$, which can be removed by taking the minimum with respect to it. Thus we can write that
			\begin{equation*}
			H_{\sss_N}(V_N) \geqslant \sum_{a} P_{\sss_N}[X_N = a] H(V_N|X_N =a) \geqslant D_1,
			\end{equation*}
			where we set
			\begin{equation}
			\begin{split}
			D_1 := \min_a\bigg[  \sum_c \min_{ij}\bigg( - P_{\sss_N}[S_{N+1}^{(i)} = a+c| S^{(j)}_N =a]\log P_{\sss_N}[S_{N+1}^{(i)} = a+c| S^{(j)}_N =a]\bigg) \bigg].
			\end{split}
			\end{equation}
			$D_1$ is a positive number, since $P_{\sss_N}[S_{N+1}^{(i)} = a+c| S^{(j)}_N =a] \in (0,1)$ (we exclude the case of \emph{deterministic} space process) and only positive probabilities contribute to the entropy, as said above. To explicitly show that the $D_1$ does not depend on $P_{\sss_N}[X_N = a]$ and $P_{\sss_N}[V_N = c]$, let us study in detail the terms $P_{\sss_N}[S_{N+1}^{(i)} = a+c| S^{(j)}_N =a]$. Recalling that the random walks are independent and that $P_{\sss_N}[S_{N+1}^{(i)} = a+c| S^{(j)}_N =a] \neq 0$ only for the $S^{(j)}_N \in \sss_N$, we can write that
			\begin{equation*}
			P_{\sss_N}[S_{N+1}^{(i)} = a+c| S^{(j)}_N =a] =
			\begin{cases}
			0 \mbox{ if $i = j$ and $c \neq \pm 1$;} \\
			p_i \mbox{ if $i = j$ and $c = 1$;} \\
			1 - p_i \mbox{ if $i = j$ and $c = - 1$;} \\
			P_{\sss_N}[S_{N+1}^{(i)} = a+c] \mbox{ if $i \neq j$}
			\end{cases}
			\end{equation*}
			where $p_i$ and $1 - p_i$ are the transition probabilities of the $i$-th random walk, which are fixed by hypothesis. Again, the first case is excluded since only positive probabilities contribute to the entropy. What we need to check is the last case, namely $P_{\sss_N}[S_{N+1}^{(i)} = a+c] $. Since in the configuration $\sss_N$ there is also the $i$-th random walk, this term reduces to
			\begin{equation*}
			P_{\sss_N}[S_{N+1}^{(i)} = a+c] = P[S^{(i)}_{N+1} =a+c | S^{(i)}_N = e]
			\end{equation*}
			for some $e \in \Int$. The only terms of this kind that contribute to the entropy are those having $e = a + c \pm 1$, i.e. the transition probabilities of the $i$-th random walks, which are fixed by hypothesis. Thus fixing $p_i$ for all $i \in I$ implies that $D_1$ is a positive constant. Summarising we showed that
			\begin{equation*}
			H_{\sss_N}(X_N) + H_{\sss_N}(V_N) \geqslant D_1,
			\end{equation*}
			when we vary over any possible value of $P_{\sss_N}[X_N =a]$ and when the transition probabilities of the $M$ random walks are fixed.\newline
			
			To conclude the proof we need to study what happens when we vary over all possible values of $P_{\sss_N}[V_N =c]$. Similarly to the previous case, $H_{\sss_N}(V_N) \geqslant 0$ while $H_{\sss_N}(X_N)$ changes according with the inequality
			\begin{equation*}
			H_{\sss_N}(V_N) \geqslant \sum_c P_{\sss_N}[V_N = c] H_{\sss_N}(X_N | V_N = c),
			\end{equation*}
			where $H_{\sss_N}(X_N|V_N =c)$ is the entropy computed using $\alpha_{\sss_N}(a,c) = P_{\sss_N}[X_N =a | V_N =c]$. From the definition  of $X_N$ and $V_N$, one can conclude that
			\begin{equation*}
			\{X_N = a\} \cap \{V_N=c\} = \{X_N = a\} \cap \{V_N=c\} \cap \{X_{N+1} = a+c\} = \{V_N=c\} \cap \{X_{N+1} = a+c\}.
			\end{equation*}
			This implies that $P_{\sss_N}[X_N =a, V_N =c] = P_{\sss_N}[X_{N+1} = a+c, V_N =c]$, i.e.
			\begin{equation*}
			\begin{split}
			\alpha_{\sss_N}(a,c) = P_{\sss_N}[X_N =a | V_N =c] = \frac{P_{\sss_N}[X_N =a, V_N =c]}{P_{\sss_N}[V_N =c]} = \frac{P_{\sss_N}[X_{N+1} = a+c, V_N =c]}{P_{\sss_N}[V_N =c]}.
			\end{split}
			\end{equation*}
			As before, the whole analysis reduces to the study of this term. Given $P_{\sss_N}[X_{N+1} = a+c, V_N = c]$ we can write that
			\begin{equation*}
			\begin{split}
			& P_{\sss_N}[X_{N+1} = a+c, V_N =c] = \sum_{i = 1}^M P_{\sss_N}[X_{N+1} = S_{N+1}^{(i)}| S_{N}^{(i)} = a+c, V_N =c] P[S_{N+1}^{(i)} =a+c, V_N =c] \\
			&=\sum_{i = 1}^M P_{\sss_N}[X_{N+1} = S_{N+1}^{(i)}| S_{N+1}^{(i)} = a+c, V_N =c] \bigg(\sum_{j,d} P_{\sss_N}[S_{N+1}^{(i)} =a+c| V_N =c,I_{N} = j,X_{N+1} =d] \cdot \\
			&\mspace{580mu}\cdot P_{\sss_N}[V_N =c,I_{N} = j,X_{N+1} =d]\bigg) \\
			&= \sum_{i = 1}^M P_{\sss_N}[X_{N+1} = S_{N+1}^{(i)}| S_{N+1}^{(i)} = a+c, V_N =c] \bigg(\sum_{j,d} P_{\sss_N}[S_{N+1}^{(i)} = a+c| V_N =c,I_{N} = j,X_{N+1} =d] \cdot \\
			&\mspace{500mu}P_{\sss_N}[I_{N} = j,X_{N+1} =d|V_N =c]  P_{\sss_N}[V_N = c]\bigg).
			\end{split}
			\end{equation*}
			Observing that the event $\{V_N =c\}\cap\{I_{N} = j\}\cap\{X_{N+1} =d\} = \{S^{(j)}_{N} = d - c\}$, we conclude that
			\begin{equation*}
			\begin{split}
			\alpha_{\sss_N}(a,c) =\sum_{i = 1}^M P_{\sss_N}[X_{N+1} = S_{N+1}^{(i)}| S_{N+1}^{(i)} = a+c, V_N =c] \bigg(\sum_{j,d} P_{\sss_N}[S_{N+1}^{(i)} = a+c| S^{(j)}_{N} = d-c] \cdot\\
			\cdot P_{\sss_N}[I_{N} = j,X_{N+1} =d|V_N =c]\bigg).
			\end{split}
			\end{equation*}
			Note that
			\begin{equation}\label{lemma1bis}
			\begin{split}
			\sum_{i = 1}^M P_{\sss_N} & [X_{N+1} = S_{N+1}^{(i)}| S_{N+1}^{(i)} = a+c, V_N =c] = 1 \\
			&\sum_{j,d} P_{\sss_N}[I_{N } = j,X_{N+1} =d|V_N =c] =1
			\end{split}
			\end{equation}
			Defining $\tilde{\gamma}(i):= P_{\sss_N}[X_{N+1} = S_{N+1}^{(i)}| S_{N+1}^{(i)} = a+c, V_N =c]$ and $\tilde{\eta}(i,j) := P_{\sss_N}[S_{N+1}^{(i)} = a+c| S^{(j)}_{N} = d-c]$, the whole analysis done in the previous case can be repeated. One has simply to replace $\gamma(i)$ with $\tilde{\gamma}(i)$, $\eta(i,j)$ with $\tilde{\eta}(i,j)$ and use \eqref{lemma1bis} instead of \eqref{lemma1}, obtaining
			\begin{equation*}
			H_{\sss_N}(X_N|V_N =c) \geqslant \sum_a \min_{i,j,d} \left( - P_{\sss_N}[S_{N+1}^{(i)} = a+c| S^{(j)}_{N} = d-c] \log P_{\sss_N}[S_{N+1}^{(i)} = a+c| S^{(j)}_{N} = d-c] \right).
			\end{equation*}
			Setting
			\begin{equation*}
			D_2 := \min_c \sum_a \min_{i,j,d} \left( - P_{\sss_N}[S_{N+1}^{(i)} = a+c| S^{(j)}_{N} = d-c] \log P_{\sss_N}[S_{N+1}^{(i)} = a+c| S^{(j)}_{N} = d-c] \right)
			\end{equation*}
		\end{widetext}
		which is a positive constant, we conclude that
		\begin{equation*}
		H_{\sss_N}(X_N) + H_{\sss_N}(V_N) \geqslant D_2,
		\end{equation*}
		when we vary over any possible value of $P_{\sss_N}[V_N = c]$ and when the transition probabilities of the $M$ random walks are fixed. Setting $D:= \min\{D_1,D_2\}$ the statement of the theorem follows. This concludes the proof.
		
	\end{proof}
	
	We can better grasp the physical meaning of the inequality between entropies proved above, considering a particular case of space process. Assume that all the random walks of the space process are \emph{identically distributed}. This means that if $p_i = P[S_{N+1}^{(i)} = a+1| S_N^{(i)} = a]$ are the transition probabilities and $\pi^{(i)}$ are the probability distributions of the initial position of all random walks, we have
	\begin{equation*}
	\begin{split}
	p_1 &= p_2 = \cdots = p_M \\
	\pi^{(1)} &= \pi^{(2)} = \cdots = \pi^{(M)}.
	\end{split}
	\end{equation*}
	This implies that $P[S^{(i)}_{N} = a] = P[S^{(j)}_{N} = a]$ for any $a \in \Int$, for any $N\geqslant0$, and any $i,j \in \{1, \cdots, M\}$. Consider the value of the constant $D_1$. The $\min_i$ can be eliminated since all the probabilities are equals. Hence
	\begin{equation*}
	D_1 = D_2= - p\log p - (1-p) \log (1-p)
	\end{equation*}
	Thus we can conclude that $D = -p \log p -(1-p) \log (1-p)$. This is the so called binary entropy, which vanishes only if $p = 0,1$ namely if that space is a deterministic process, a case which is excluded. The physical meaning of the inequality $H_{\sss_N}(X_N) + H_{\sss_N}(V_N) \geqslant D$, in this case, is now clear: the uncertainty that we have on $X_N$ or $V_N$ must be at least  equal to the uncertainty we have on a single point in the future configurations of the space process (given that at time $N$ the configuration is $\sss_N$).
	
	\subsection{Construction of the Hilbert space structure for model A}
	
	Theorem \ref{theoV1} implies that after the removal of the space process, $X_N$ and $V_N$ are described using two distinct probability spaces if we want to use the unconditional transition probabilities $\alpha(a+c,a)$. Theorem \ref{EURmodelA} tell us that under the same assumptions, the position and the velocity of the particle in model A, fulfil an entropic uncertainty relation. At this point, we may proceed algebraically and define the smallest $C^*$-algebra which is capable to describe both $X_N$ and $V_N$ after conditioning, and the entropic uncertainty relation \eqref{EURmodelA} tells us that this algebra is non-commutative \cite{LC}. Then, we can represent these elements of the algebra as two non-commuting operators over a Hilbert space via the GNS theorem. Despite this is a legitimate way to proceed, in this section, using the results collected in \cite{LC}, we will use a more constructive approach. In particular, we show how to construct the operators associated to these random variables and how to define a suitable Hilbert space on which they are defined.
	
	Consider the position random variable $X_N$. After conditioning on a particular configuration of the space process $\sss_N$, $X_N$ can be seen as the as the following map between probability spaces
	\begin{equation*}
	\begin{CD}
	(\Omega_{I} \times \Omega_{\Spa},\E_{I} \otimes \E_{\Spa}, P)|_{\sss_N}  @>X_N>> (\Omega_{X_N}, \E_{X_N}, \mu_{X_N})
	\end{CD}
	\end{equation*}
	where $\Omega_{X_N} = X_N(\Omega_I \times \{\sss_N\})$, $\E_{X_N} = \mathcal{P}(\Omega_{X_N})$ and $\mu_{X_N}:= P_{\sss_N} \circ X_N^{-1}$. As we have seen in \cite{LC}, random variables over a probability space form a commutative von-Neumann algebra which is isomorphic to an algebra of multiplicative operators over an Hilbert space. In this particular case, the random variables over $(\Omega_{X_N}, \E_{X_N}, \mu_{X_N})$ (on which $X_N$ is represented by the identity map) form the abelian von-Neumann algebra $\mathcal{V}_c(L_2(\Omega_{X_N}, \mu_{X_N}))$. Seen as element of this algebra, the random variables over $(\Omega_{X_N}, \E_{X_N}, \mu_{X_N})$ are multiplicative operators over $L_2(\Omega_{X_N}, \mu_{X_N})$. 
	
	Similar considerations hold for the velocity of the particle. The main difference is the definition of $\Omega_{V_N}$, i.e. the set of all the elementary outcomes. It is not difficult to understand that, if we fix the space process only, $\Omega_{V_N}$ seems to contain more outcomes of those one should expect. The number of outcomes of the space process is $M^2$, i.e. $\mbox{card } \Omega_{X_N} = M^2$. This because the origin and the point of $\Spa_N$ selected by the selection process $I_N$, can take $M$ different values. For the velocity process similar considerations lead to $\mbox{card} \Omega_{V_N} = M^4$. However, we have to take into account that we cannot detect the movement of the origin: $S_{N+1}^{(i_O)} - S_N^{(i_O)}$ must be set equal to $0$, and all the situations where this does not hold must be identified with it \footnote{More precisely, we can define an equivalence relation between $X_{N+1}$ and $X_N$: $X_{N+1} \sim X_N$ if $X_{N+1} - X_N = S^{(i_o)}_{N+1} - S^{(i_o)}_N$. In this way we restrict our attention to the intrinsic motion of the particle.}. After that the velocity can takes only $M^2$ different values (the $M$'s of $S_N^{(i_{N+1})}$ times the $M$'s of $S_N^{(i_N)}$). Thus doing that we have $\mbox{card }\Omega_{V_N} = M^2$. After this observation, we may see the velocity random variable, after conditioning to $\sss_N$, as the map
	\begin{equation*}
	\begin{CD}
	(\Omega_{I} \times \Omega_{\Spa},\E_{I} \otimes \E_{\Spa}, P)|_{\sss_N}  @>V_N>> (\Omega_{V_N}, \E_{V_N}, \mu_{V_N})
	\end{CD}
	\end{equation*}
	where $\Omega_{V_N} = V_N(\Omega_I \times \{\sss_N\})$, $\E_{V_N} = \mathcal{P}(\Omega_{V_N})$ and $\mu_{V_N}:= P_{\sss_N} \circ V_N^{-1}$. Also in this case, the random variables over  $(\Omega_{V_N}, \E_{V_N}, \mu_{V_N})$, are elements of a commutative von-Neumann algebra $\mathcal{V}_c(L_2(\Omega_{V_N}, \mu_{V_N}))$ (i.e. multiplicative operators on $L_2(\Omega_{V_N}, \mu_{V_N})$).\newline
	
	Thus both $X_N$ and $V_N$ can be represented by multiplicative operators on suitable Hilbert spaces. Note that the two Hilbert spaces are different and depend on the probability measure. In order to construct a common Hilbert space on which both operators are defined, we should invoke the spectral representation theorem, as explained in \cite{LC}. We recall that the spectral decomposition theorem tells that, given an operator $\hat{T}$, there exist a surjective isometry $\hat{U}_i: \Hi_i \rightarrow L_2(\sigma(\hat{T}), \mu_i)$ such that $\hat{U}_i^* \hat{T}|_{\Hi_i} \hat{U}_i$ is a multiplicative operator on $L_2(\sigma(\hat{T}), \mu_i)$, i.e. an element of $\mathcal{V}_c(L_2(\sigma(\hat{T}), \mu_i))$. Consider the position random variable $X_N$. We know that it is a multiplicative operator on $L_2(\Omega_{X_N}, \mu_{X_N})$, and let us now choose to parametrise the probability measure of the position random variable with the outcome of $X_N$. This can be achieved in the following way. Take $a \in \Omega_{X_N}$ and consider the probability measure $P_{\sss_N}^{(a)}$, which is defined such that $\mu_{X_N}(c) = P_{\sss_N}^{(a)} \circ X_N^{-1}(c) = \delta_{a,c}$. We can parametrise the probability measure of $X_N$ with its outcomes defining $\mu_{X_N|a} :=P_{\sss_N}^{(a)} \circ X_N^{-1}$. Doing that we obtain a collection of Hilbert spaces $\{L_2(\Omega_{X_N}, \mu_{X_N|a})\}_{a \in \Omega_{X_N}}$. Now, the random variable $X_N$ can be represented with an operator $\hat{X}_N$, having spectrum $\sigma(\hat{X}_N) = \Omega_{X_N}$. The spectral decomposition theorem tells that there exists a collection of Hilbert spaces $\{\Hi_a\}_{a \in \Omega_{X_N}}$ and surjective isometries $\hat{U}_a : \Hi_a \rightarrow L_2(\Omega_{X_N}, \mu_{X_N|a})$, which allows to define the Hilbert space 
	\begin{equation*}
	\Hi (X_N) := \bigoplus_{a \in \Omega_{X_N}} \Hi_a
	\end{equation*}
	on which $\hat{X}_N$ can be seen as a multiplicative operator. The spectral representation theorem tells that if $\{|x_N \rangle\}$ is a basis of $\Hi(X_N)$ such that $| x_N \rangle \in \Hi_{x_N}$ for any $x_N \in \Omega_{X_N}$, then $X_N$ can be represented by the operator
	\begin{equation*}
	\hat{X}_N = \sum_{x_N \in \Omega_{X_N}} x_N |x_N \rangle \langle x_N |.
	\end{equation*}
	With similar considerations, for $V_N$ we obtain
	\begin{equation*}
	\Hi(V_N) : = \bigoplus_{c \in \Omega_{V_N}} \Hi_c
	\end{equation*}
	on which the operator $\hat{V}_N$ representing the velocity random variable, is diagonal
	\begin{equation}\label{velopMA}
	\hat{V}_N = \sum_{v_N \in \Omega_{V_N}} v_N |v_N \rangle \langle v_N |.
	\end{equation}
	At this point, we impose the condition
	\begin{equation*}
	\Hi(X_N) = \Hi(V_N)
	\end{equation*}
	i.e. that the two Hilbert spaces are \emph{unitary equivalent}. This is possible since the dimension of both Hilbert spaces is $M^2$: both Hilbert spaces are constructed from the spectrum of $\hat{X}_N$ or $\hat{V}_N$, and both have the same number of elements. Since Hilbert spaces of equal dimension are always isomorphic, there exists a unitary map between them, i.e. there exists
	\begin{equation*}
	\hat{U}: \Hi(V_N) \rightarrow \Hi(X_N),
	\end{equation*}
	such that $\hat{U}\hat{U}^* = \Id_{\Hi(X_N)}$ and $\hat{U}^*\hat{U} = \Id_{\Hi(V_N)}$.  This unitary mapping allows to have, on the same Hilbert space, the operators representing the position and the velocity random variables. More precisely, take the velocity operator $\hat{V}_N$ on $\Hi(V_N)$ defined in \eqref{velopMA}, then the unitary map mentioned above allows us to write 
	\begin{equation*}
	\begin{split}
	\hat{V}_N|_{\Hi(X_N)} &= \hat{U} \bigg(\sum_{v \in \Omega_{V_N}} v_N |v_N \rangle \langle v_N | \bigg) \hat{U}^* \\
	&= \sum_{v_N \in \Omega_{V_N}} v_N \hat{U}|v_N \rangle \langle v_N |\hat{U}^*,
	\end{split}
	\end{equation*}
	which represents the velocity random variable on $\Hi(X_N)$, the Hilbert space constructed from the spectrum of the position operator (on which $\hat{X}_N$ is diagonal). The entropic uncertainty relation, ensures that $X_N$ and $V_N$ as operators on the same Hilbert space, do not commute. In fact, it implies \cite{maassen1988generalized}
	\begin{equation}\label{Massenbound}
	\max_{x_N, v_N}|\langle x_N | v_N \rangle| \leqslant e^{- \frac{D}{2}} < 1, 
	\end{equation}
	as already observed in \cite{LC}. Thus the two operators cannot be diagonalised on the same basis, i.e. they do not commute. 
	
	We can also represent on $\Hi(X_N)$ the velocity random variable directly. Indeed on this Hilbert space, we may always consider a generic basis $\{ |w_N \rangle\}_{w_N \in \Omega_{V_N}}$ and impose that $\hat{V}_N$ is diagonal on this basis, i.e.
	\begin{equation*}
	\hat{V}_N = \sum_{w_N \in \Omega_{V_N}} w_N | w_N \rangle \langle w_N |.
	\end{equation*}
	We can always parametrise the probability measure of $\mu_{V_N}$ using the outcome of $X_N$ simply defining $\mu_{V_N|a}:= P_{\sss_N}^{(a)} \circ V_N^{-1}$. Then we obtain the collection of Hilbert spaces $\{L_2(\Omega_{V_N}, \mu_{V_N|a})\}_{a \in \Omega_{X_N}}$. For a given $a \in \Omega_{X_N}$, the entropic uncertainty relation \eqref{EUR-ModA} forbids to have delta-like probability measure for both operators. Indeed, considering $\hat{X}_N$, we have
	\begin{equation*}
	\begin{split}
	\mu_{X_N| a} := \langle \psi  | \hat{P}_{\Hi_{x_N}} \psi \rangle = \langle \psi| x_N \rangle \langle x_N | \psi \rangle = \delta_{x_N,a}
	\end{split}
	\end{equation*}
	which is possible only if $| \psi \rangle = | a \rangle$.  On the other hand for $\hat{V}_N$ , if $\hat{P}_{\Hi_{w_N}}$ is the projector on the subspace of $\Hi(X_N)$ associated to the eigenvalue $w_N$ (i.e. the outcome $w_N$ of the random variable $V_N$), we have
	\begin{equation*}
	\begin{split}
	\mu_{V_N|a} :&= \langle \psi |\hat{P}_{\Hi_{w_N}} \psi \rangle = \langle a | w_N \rangle \langle w_N | a \rangle = |\langle a | w_N \rangle |^2.
	\end{split}
	\end{equation*}
	Note that under the symmetry condition \eqref{transiprobi} on the unconditional transition probabilities, i.e. $P[V_{N} = c| X_N = a] = P[X_N = a| V_{N} = c]$,  the $\mu_{V_N|a}$ probabilities are consistent with usual interpretation of transition probabilities in quantum theory. Since the entropic uncertainty relation hold, \eqref{Massenbound} forbids that $|w_N \rangle$ and $|x \rangle$ to be orthogonal. Again, we conclude that $X_N$ and $V_N$ can be represented on a common Hilbert space, $\Hi(X_N)$, using two operators $\hat{X}_N$ and $\hat{V}_N$ which cannot be diagonalised on the same basis. Note that this $\hat{V}_N$ coincides exactly with $\hat{V}_N|_{\Hi(X_N)}$ thanks to the existence of the unitary map $\hat{U}: \Hi(V_N) \rightarrow \Hi(X_N)$. Clearly, also $X_N$ can be represented on $\Hi(V_N)$ directly, following a similar procedure. In this sense the whole description is consistent: starting the construction of the Hilbert space from $X_N$ or $V_N$ does not change anything, as it should be.
	
	Finally, we conclude by observing that the probabilistic content is now encoded in the vectors $|\psi\rangle$ of the constructed Hilbert spaces. In fact, given $|\psi \rangle \in \Hi(X_N)$ (or $\Hi(V_N)$), we can write that
	\begin{equation*}
	\begin{split}
	\Ex_{\sss_N}[V_N] &= \Tr{|\psi \rangle\langle \psi | \hat{V}_N} = \sum_v v \Tr{|\psi \rangle\langle \psi |v \rangle\langle v |} 
	\end{split}
	\end{equation*}
	where $ \Tr{|\psi \rangle\langle \psi |v \rangle\langle v |} = P_{\sss_N}[V_N = v]$ is the probability distribution for $V_N$ after conditioning. The probability distribution for $V_N$ can be related with the distribution of $X_N$ as follows
	\begin{equation*}
	\begin{split}
	P_{\sss_N}[V_N = v] &= \sum_x \langle x |\psi \rangle\langle \psi |v \rangle\langle v | x \rangle \\
	&= \sum_x \sum_{x'} \langle x |\psi \rangle\langle \psi |x' \rangle\langle x' |v \rangle\langle v | x \rangle \\
	&= \sum_x  \alpha(x+v,x)P_{\sss_N}[X_N = x]  \\ 
	&\mspace{20mu}+ \sum_{x \neq x'} \langle x |\psi \rangle\langle \psi |x' \rangle\langle x' |v \rangle\langle v | x \rangle 
	\end{split}
	\end{equation*}
	where we used $P_{\sss_N}[X_N = x] = |\langle x | \psi \rangle|^2$ and $\alpha(x+v,x) = |\langle x | v \rangle|^2$. Note that the second term in the last sum (the interference term) corresponds to the correction term $\delta(V_N|X_N,\sss_N)$ in theorem \ref{theoV1}. However, the method used here does not provide a way to determine uniquely the objects on $\Hi(X_N)$ (or $\Hi(V_N)$) associated to a given set of probability distributions, as already noted. In fact, the method proposed does not provide an explicit way to compute the phase of $\langle x_N | v_N \rangle$ starting from the interference term. However QRLA may indicate a possible way to do that \cite{khrennikov2005interference,khrennikov2009contextual}

	\subsection{Final remarks and main limitations of model A}\label{sec3f}
	
	Let us conclude our presentation of model A, with some observations and a discussion of some limitations of the model. 
	
	The model is surely interesting because it is capable to derive non-commuting operators over an Hilbert space starting from a \textquotedblleft classical" description (in probabilistic sense). Such non-commutativity, at least mathematically, seems to be related to the definition used for the two random variables of the particle process. So despite they seem reasonable definitions, we should at least argue why they seem to be related to the position and momentum operator in ordinary (non-relativistic) quantum theory. In particular, once we choose to describe a quantum particle in $L_2(\Rea^n)$ and we decide that the symmetry group of non-relativistic physics is the Galilean group, by the Stone-von Neumann theorem, we can justify that the position and momentum operator are defined as
	\begin{equation*}
	\hat{X}_i \psi(\mathbf{x}) = x_i \psi(\mathbf{x}) \mspace{50mu} \hat{P}_i\psi(\mathbf{x}) = - i\frac{\partial}{\partial x_i}\psi(\mathbf{x}),
	\end{equation*}
	for $\psi(\mathbf{x}) \in \mathcal{S}(\Rea^n)$.  Consider the 1-D case only. In quantum mechanics, the classical relation between position and momentum for a point-like particle ($p = m\dot{x}$) does not seem a priori valid. Nevertheless from the Ehrenfest theorem, we have that
	\begin{equation*}
	\langle P_t \rangle = m \frac{d}{dt} \langle X_t \rangle.
	\end{equation*}
	Assume that for some reason time is discrete (for example because the limited accuracy of the clock). The limit now can be replaced by an inferior and, using the Wigner quasi-probability distribution $W(x,p)$, we can write that
	\begin{equation*}
	P_{t} = m \inf_{\delta t} \frac{X_{t + \delta t} -X_{t}}{\delta t} \mspace{50mu} W(x,p)\mbox{-a.s.}.
	\end{equation*}
	Setting $\delta t =1$ ($\delta t$ is our unit of time) and $t = N\delta t$, we can write $P_N = m (X_{N+1} - X_N) = mV_N$. This consideration justifies, at least at the qualitative level, the use of the  two random variables described in model A. 
	An interesting feature of the model presented here is that the square of the number of points of the space process (which is removed) is equal to the dimension of the (minimal) Hilbert space on which we represent the particle process, namely $X_N$ and $V_N$. Thus, the dimension of the Hilbert space in model A seems to encode information on the removed process, in this model. We do not expect this feature to be fundamental (the Hilbert space structure seems more related to probabilistic rather than \textquotedblleft geometrical" considerations) but this observation will be useful in future.
	
	A limitation of the model is that time is treated as a discrete parameter, a choice that does not allow to compare the model directly with non-relativistic quantum mechanics. In addition, also space appears discrete, in the sense that it can take values only over a lattice and not on the whole $\Rea$. Summarising we can neither derive the commutation relation \eqref{[Q,P]}, nor attempt a comparison with non-relativistic quantum mechanics using Model A, because of the following: $1)$ the Hilbert space of model A is finite dimensional, $2)$ $\hat{X}_N$ and $\hat{V}_N$ are bounded operators with discrete spectrum  and, $3)$ time is discrete. However, we was able to successfully represent position and velocity (momentum) of the particle as non-commuting operators over a common Hilbert space, a key feature of non-relativistic quantum mechanics.

	\section{Model B: Continuous-time 1-D kinematics on a random space}\label{mB}\label{Part4}
	
	We have shown in part \ref{mA} that model A exhibits very interesting features from the point of view of quantum mechanics. Nevertheless, it also has some limitations: time is a discrete parameter and the spectrum of the position operator $\hat{X}_N$ is discrete. They do not allow for a direct comparison with ordinary quantum mechanics. To allow for this comparison, one may try to generalize Model A to continuous-time random variables. Here we will show how to do it.

	\subsection{The space process}\label{Model B - space process}\label{sec4a}
	
	In order to generalize model A to the continuous time case, we may start by generalizing the space process. Instead of considering the space process as a collection of random walks, we may consider their \textquotedblleft continuous limits", i.e. Wiener processes. Let us recap the basic features of the Wiener process \cite{klebaner2005introduction}, as done for the random walk. A \emph{ Wiener process $W_t$ starting at $y$} is a Gaussian process with mean $\Ex[W_t]=y$,  and covariance $\Ex[W_tW_s] =\min(t,s)$. This is one of the possible equivalent definitions of a Wiener process, and it implies that (in the 1D case)
	\begin{equation*}
	P[W_t \in A] = \int_A \frac{1}{\sqrt{2\pi t}} e^{-\frac{(x-y)^2}{2t}}dx.
	\end{equation*}
	As consequence of its definition, the Wiener process $W_t$ is a continuous function of the parameter $t$ for all $t \in \Rea^+$, in the sense that there exists always a \emph{continuous version} of the Wiener process (with \textquotedblleft version of a process $X_t$ " we mean that there exists another process $Y_t$ such that $P[X_t=Y_t] = 1$ for any $t \in \Rea^+$, i.e. the two processes are statistically indistinguishable). For a Wiener process, the trajectories (which can be thought as the function $\omega(t) := W_t(\omega)$) have the following properties: $i)$ they are nowhere differentiable; $ii)$ they are never monotone; $iii)$ they have infinite variation in any interval; $iv)$ they have quadratic variation equal to $t$ in the interval $[0,t]$. More generally, let $C(\Rea^+,\Rea)$ be the space of all functions $t \mapsto f_t$ taking value on $\Rea$ and continuous for any $t \in \Rea^+$. $C(\Rea^+,\Rea)$ can be equipped with a norm, which allows to define open sets (i.e. a topology). As usual these open sets can be used to construct a Borel $\sigma$-algebra on $C(\Rea^+,\Rea)$, say $\mathscr{B}(C(\Rea^+,\Rea))$.  The Wiener process can be seen as the identity function on $(C(\Rea^+,\Rea), \mathscr{B}(C(\Rea^+,\Rea), \gamma)$ where $\gamma$ is the so called \emph{Wiener measure}. The set of all continuous functions $f_t \in C(\Rea^+,\Rea)$ which does not fulfil $i)$ -- $iv)$ have zero measure under $\gamma$. Such a probability space is called \emph{Wiener space}. Finally we conclude by observing that if also the starting position $y$ is a random variable with distribution $\pi(dy)$ over $\Rea$, then
	\begin{equation*}
	P[W_t \in A] = \int_A \int_{\Rea}\frac{1}{\sqrt{2\pi t}} e^{-\frac{(x-y)^2}{2t}} \pi(dy) dx.
	\end{equation*}
	
	Let us now consider the space process for this model. As assumed for model A, space is discrete and evolves with time. In particular, we have the following preliminary definition which generalizes the one given for model A.
	\begin{definition}
		Let $\{W_t^{(i)}\}_{i \in I}$ be a collection of independent Wiener processes, where $|I| = M \in \Nat$. Such collection will be called \emph{space process} for model B.
	\end{definition}
	We will label this process by $\Spb$. At any given time $t \in \Rea^+$, the space process is a collection of $M$ points on $\Rea$, which are the positions of the $M$ Wiener processes: in this sense the space is discrete and evolves, in a continuous way, in time. Also in this case we may have two possible descriptions of the space: one is to consider the points with respect to the real line (we will choose this point of view, as done for model A), while the second is to describe the effects of the time evolution from the point of view of the \textquotedblleft ordering among the points" (as explained in appendix B). Because of independence, equation \eqref{Pro1} holds true if we simply substitute $P[S_N^{(i)} = s_{i}]$ with $P[W_t^{(i)} \in A_{i}]$, where $A_i \subset \Rea$ for any $i \in \{1,\cdots,M\}$, and similarly for \eqref{Pro2} and generalisation. Because $P[W_t^{(i)} \in A_{i}]$ can be written as the integral over $A_i$ with respect to a probability density $\rho_{W_t^{(i)}}(x_i)$, equation \eqref{Pro1} is replaced by the following
	\begin{equation}\label{Pro1New}
	\rho_{\Spb_t}(\mathbf{S}_t) = \prod_{i = 1}^M \rho_{W_t^{(i)}}(x_{i}),
	\end{equation}
	where $\rho_{\Spb_t}(\mathbf{S}_t)$ is the probability density of the probability measure $\tilde{P}^B[ \Spb_t \in A ]$. In a similar way one can generalise \eqref{Pro2} and any other density for the space process.  At this point, as done for model A, we may give the following definition for the space process.
	\begin{definition}
		Let $\{W_t^{(i)}\}_{i \in I}$ be  a collection of $M = |I| \in \Nat$ Wiener processes defined on the Wiener spaces $\{(\Omega_i,\E_i,P_i)\}_{i \in I}$. Let us define
		\begin{enumerate}
			\item[i)] $\Omega_{\Spb} :=\Omega_1 \times \cdots \times \Omega_M$;
			\item[ii)] $\E_{\Spb}$ is the Borel $\sigma$-algebra generated by the open sets of $\Omega_{\Spb}$\footnote{To define an open set on $\Omega_{\Spb}$ we may use the topology induced by the norm $\|\sss\| := \sup_{t \in [0,T]} |\sss_t|_M$, where $| \cdot |_M$ is the $M$-dimensional euclidean norm. This is what is typically done on Wiener spaces.};
			\item[iii)] $\tilde{P}^B: \E_{\Spb} \rightarrow [0,1]$ defined from the $\{P_i\}_{i \in I}$, via the densities as in \eqref{Pro1New} and generalisations.
		\end{enumerate}
		The space process is the stochastic process on $(\Omega_{\Spb},\E_{\Spb},\tilde{P}^B)$ defined as the identity function, namely $\Spb(\omega_1,\cdots,\omega_M ) = (\omega_1, \cdots, \omega_M)$. 
	\end{definition}
	The set of all possible configurations of the space process at time $t$ will be labeled by the symbol $\mathcal{S}(t)$. This completes our description for the space process in model B.
	
	\subsection{The particle process}
	
	Again, a particle is considered as a point-like object. It jumps from one point of space to another and it is completely characterized by the \emph{position} and \emph{velocity} random variables.
	
	The \emph{position random variable}, labeled by $X_t$,  is interpreted as the actual position of the particle at time $t$ with respect to a chosen origin. Hence, if $(\Omega_{\Spb},\E_{\Spb},\tilde{P}^B)$ is the probability space of the space process, $(\Omega_{I},\E_{I},P_{I})$ is a probability space on which an integer value stochastic process $I_t:\Omega_I \rightarrow \{1,\cdots,M\}$ is defined (called \emph{section process}), and $W_t^{(i_0)}$ is a chosen origin of a reference frame on $\Spb_t$, then
	\begin{equation}\label{PosProcCont}
	X_t(\omega_{X}) := \pi_{I_t(\omega_{I})}(\Spb_t(\sss)) - W_t^{(i_0)},
	\end{equation}
	where $\pi_i$ is the projector of the $i$-th component of an $M$-tuple, and $\omega_{X} = (\omega_{I}, \sss)$ with $\omega_{I} \in \Omega_{I}$ and $\sss \in \Omega_{\Spb}$. Thus we have the following definition.
	\begin{definition}\label{ModelB:Xdef}
		Consider the probability space $(\Omega_I \times \Omega_{\Spb}, \E_I \otimes \E_{\Spb}, P^B)$ and a measurable space $(\Rea,\borel)$. The random variable $X_t$ is the $\borel$-measurable function
		\begin{equation*}
		X_t: \Omega_I \times \Omega_{\Spb} \rightarrow \Rea
		\end{equation*}
		defined as in \eqref{PosProcCont}. $X_t$ represents the position of the particle at time $t$.
	\end{definition}
	Clearly, as any random variable $X_t$ induces a probability distribution $\mu_{X_t} = P^B \circ X_t^{-1}$ and, on the probability space $(\Rea,\borel,\mu_{X_t})$ it can be considered as the identity function. Also in this case the space process can be described on $(\Omega_I \times \Omega_{\Spb}, \E_I \otimes \E_{\Spb}, P^B)$, by simply demanding that $P^B \circ [\Spb]^{-1} = \tilde{P}^B$. Again if no confusion arises, we omit the suffix $^B$ in the probability measure $P^B$. 
	
	In model B, the particle moves by jumps from one point to another. This time the frequency of the jumps is assumed to be infinite, which means that the particle jumps from one point to another at each instant of time. In this way, we can say that it is the continuous time generalization of the kinematics described in model A. We do not generalize the definition of the velocity process given before directly. This time we use the following definition:
	\begin{equation}\label{ContVel}
	V_t(t'):= \frac{X_{t'} - X_t}{t' - t}
	\end{equation}
	where we always assume $t' > t$. More formally we adopt the following definition for $V_t(t')$.
	\begin{definition}\label{ModelB:Vdef}
		Consider the probability space $(\Omega_{I} \times \Omega_{\Spb}, \E_{I} \otimes \E_{\Spb}, P^B)$ and a measurable space $(\Rea,\borel)$. Let $t,t' \in \Rea$ such that $t' > t$, the random variable $V_t(t')$ is the $\borel$-measurable function
		\begin{equation*}
		V_t(t'): \Omega_I \times \Omega_{\Spb} \rightarrow \Rea
		\end{equation*}
		defined in \eqref{ContVel}. $V_t(t')$ represents the mean velocity of the particle in the interval $[t,t']$.
	\end{definition}
	Also in this case $V_t(t')$ can be seen as the identity random variable on the probability space $(\Rea,\borel,\mu_{V_t(t')})$, where $\mu_{V_t(t')} = P^B \circ V_t(t')^{-1}$. As in model A, for the description of the particle we need to introduce the transition probabilities. These allow to write that
	\begin{equation}\label{VelProbModelB}
	\mu_{V_t(t')}(v) = \int_{\Rea} \alpha(v,x;t') \mu_{X_t}(x)dx
	\end{equation}
	where $\alpha(v,x;t')$ are the probability densities of $V_t(t')$ \emph{given the event} $\{X_t = x\}$. Note that they depend also on the value of $t'$ used to define $V_t(t')$. Also in this case we assume that they are \emph{symmetric} under the exchange of their arguments, namely $\alpha(v,x;t') = \alpha(x,v;t')$, where $\alpha(x,v;t')$ is the probability density of $X_t$ \emph{given th event} $\{V_t(t') = v\}$. Note that this expression is nothing but the Bayes theorem for continuous random variables (see appendix C, Prop. \ref{PropCondProb1}). In what follows we will omit $t'$ in $\alpha(v,x;t')$ and $\alpha(x,v;t')$ if no confusion arises.
	
	We conclude this section defining the particle process for this model. With the definitions \ref{ModelB:Xdef} and \ref{ModelB:Vdef} we gave a meaning to the position and velocity \emph{random variables}. However, they are parametrised by $t$, the physical time which we choose to treat as an external parameter. Since we have a collection of random variables parametrised by $t$, we can speak of \emph{position and velocity stochastic process}, but with the same caution explained in section \ref{sec3b}. 
	\begin{definition}
		Let $X_t$ and $V_t(t')$ be the position and velocity process. The couple $\mathscr{P}_t(t') = (X_t,V_t(t'))$ is called \emph{particle process} of model B.
	\end{definition}
	An example of particle process over a space process is drawn in figure \ref{fig:PPB}.
	\begin{figure}[h!]
		\includegraphics[scale=0.35]{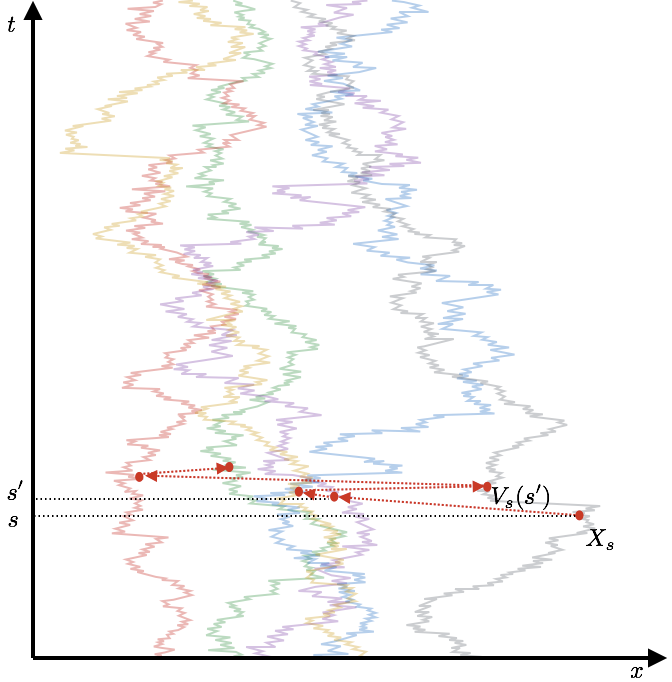}
		\caption{The particle process in model B is drawn in red. The position of the particle at a given time is given by the red point while its velocity at the same time is represented by the outgoing arrow. On the back, a possible realization of an $M=6$ space process.}
		\label{fig:PPB}
	\end{figure} 
	
	\subsection{The removal of the space process}
	
	The removal of the space process in model B is done exactly as before:
	\begin{enumerate}
		\item[i)] We consider only the conditional probability densities $\mu_{X_t|\sss_t}(x)$ and $\mu_{V_t(t')|\sss_t}(v)$ for the random variables $X_t$ and $V_t(t')$;
		\item[ii)] We fix the transition probabilities of the single point of space, i.e. we fix the transition probabilities $p^{(i)}(x,t'; y,t) := \rho_{W_{t'}^{(i)}|W_t^{(i)} = y}(x) $ of all the $M$ Wiener processes.
	\end{enumerate}
	
	As in Model A, requirement $i)$ implies that we will always work with the densities $\mu_{X_t|\sss_t}(x)$ and $\mu_{V_t(t')|\sss_t}(v)$, namely the probability distributions for $X_t$ and $V_t(t')$ \emph{given the event} $\{\Spb_t = \sss_t\}$. Clearly, we can define a joint probability space for $X_t$ and $V_t(t')$ after conditioning on $\{\Spb_t = \sss_t\}$ and, on this joint probability space, some \emph{conditional} transition probabilities $\alpha_{\sss_t}(v,x)$ can be defined. However, if we insist in using the \emph{unconditional probability density} $\alpha(v,x)$, no joint probability space can be defined. Indeed, we can prove the analogue of theorem \ref{theoV1}. Below we will prove the theorem in a slightly more general setting of what we need later: the simple case of absolute continuous measure with respect to Lebesgue will be discussed as an example later.
	\begin{theorem}\label{ContV1}
		Let $(\Omega_{I} \times \Omega_{\Spb}, \E_{I} \otimes \E_{\Spb},P)$ be the probability space on which $\Spb_t$, $X_t$ and $V_t(t')$ are defined. Assume that the probability spaces for each of these random variables has the regular conditional probability property. Then 
		\begin{equation}
		\begin{split}
		P_{\sss_t}[V_t(t') \in B ] = &\int_B \int_{\Omega_{X_t}} \alpha(v,x) P_{\sss_t}[X_t \in dx] \\
		&+ \delta(B|X_t,\sss_t)
		\end{split}
		\end{equation}
		where $\delta(B|X_t,\sss_t)$ is the Radon-Nikodym derivative with respect to $P \circ[\Spb_t]^{-1} = \mu_{\Spb_t}$, of the measure $\Gamma \mapsto \Delta(B| X_t, \Gamma )$ defined as\newline\newline\newline
		\begin{equation}\label{DeltaCont}
		\begin{split}
		\Delta(B| X_t, \Gamma ) := &\int_{\mathcal{S}(t) / \Gamma} \bigg[\int_{\Omega_{X_t}} P[V_t(t')\in B|X_t = x] \cdot \\ 
		&\cdot\mu_{X_t|\sss_t'}(dx) - P_{\sss_t'}[V_t(t')\in B] \bigg]\mu_{\Spb_t}(d\sss_t').
		\end{split}
		\end{equation}
		when $\sss_t \in \Gamma \subset \mathcal{S}_p^\epsilon (t)$.
	\end{theorem}
	\begin{proof}
		Since $\{V_t(t') \in B\} = \{ V_t(t') \in B \} \cup \{ \Spb_t \in \mathcal{S}(t) \}$, using the regular conditional probability property of probability space, we have
		\begin{equation*}
		\begin{split}
		P[&\{V_t(t') \in B\} \cup \{ \Spb_t \in \mathcal{S}(t) \} ]  
		= P[V_t(t') \in B] \\
		&=\Ex [\chi_{\{V_t(t') \in B\}}(\omega)] \\
		&= \int_{\mathcal{S}(t)} P_{\sss_t'} [V_t(t') \in B] (P \circ[\Spb_t]^{-1})(d\sss_t')
		\end{split}
		\end{equation*}
		where $P_{\sss_t'} [V_t(t') \in B]$ is the regular conditional probability $P[V_t(t') \in B| \Spb_t = \sss_t']$. Take $\Gamma \subset \mathcal{S}(t)$ and set $\mu_{\Spb}(d\sss_t') = (P \circ[\Spb_t]^{-1})(d\sss_t')$, then we can write
		\begin{equation*}
		\begin{split}
		P[V_t(t') \in B] &= \int_{\Gamma} P_{\sss_t'} [V_t(t') \in B] \mu_{\Spb}(d\sss_t') \\
		&+\int_{\mathcal{S}(t) / \Gamma} P_{\sss_t'} [V_t(t') \in B] \mu_{\Spb}(d\sss_t').
		\end{split}
		\end{equation*}
		Consider now the random variable $X_t$ and the $\sigma$-algebra $\E_{X_t}$. Then, by the law of conditional expectation, we can write that
		\begin{equation*}
		\begin{split}
		P[&\{V_t(t') \in B\} \cup \{ \Spb_t \in \mathcal{S}(t) \} ] = P[V_t(t') \in B] \\
		&= \Ex [\chi_{\{V_t(t') \in B\}}(\omega)] \\
		&= \Ex [ \Ex[\chi_{\{V_t(t') \in B\}}(\omega)| \E_{X_t}] ] \\
		&= \int_{\mathcal{S}(t)} \int_{\Omega_{X_t}}P[V_t(t')\in B|X_t = x] \mu_{X_t|\sss_t'}(dx) \mu_{\Spb_t}(d\sss_t')
		\end{split}
		\end{equation*}
		where we used the fact that the conditional expectation $\Ex[\chi_{\{V_t(t') \in B\}}(\omega)| \E_{X_t}]$ is a random variable and the disintegration theorem (see appendix C, Th. \ref{DisTheo}). Then, as done before, for a given $\Gamma \subset \mathcal{S}(t)$ we can write that:
		\begin{widetext}
			\begin{equation*}
			\begin{split}
			P[V_t(t') \in B] = \int_{\Gamma}\int_{\Omega_{X_t}} P[V_t(t')\in B|X_t = x] \mu_{X_t|\sss_t'}(dx) \mu_{\Spb_t}(d\sss_t') + \int_{\mathcal{S}(t) / \Gamma} \int_{\Omega_{X_t}} P[V_t(t')\in B|X_t = x] \mu_{X_t|\sss_t'}(dx) \mu_{\Spb_t}(d\sss_t').
			\end{split}
			\end{equation*}
			Comparing the two expressions found for $P[V_t(t') \in B]$ we obtain
			\begin{equation}\label{Derivation1}
			\int_{\Gamma} P_{\sss_t'} [V_t(t') \in B] \mu_{\Spb}(d\sss_t') = \int_{\Gamma} \int_{\Omega_{X_t}}P[V_t(t')\in B|X_t = x] \mu_{X_t|\sss_t'}(dx) \mu_{\Spb_t}(d\sss_t')  + \Delta(B|X_t,\Gamma) 
			\end{equation}
			with
			\begin{equation*}
			\begin{split}
			\Delta(B|X_t,\Gamma) = \int_{\mathcal{S}(t) / \Gamma} \bigg[\int_{\Omega_{X_t}} P[V_t(t')\in B|X_t = x] \mu_{X_t|\sss_t'}(dx) - P_{\sss_t'}[V_t(t')\in B] \bigg]\mu_{\Spb_t}(d\sss_t').
			\end{split}
			\end{equation*}
		\end{widetext}
		Note that the map $\Gamma \mapsto \Delta(B|X_t,\Gamma)$ is a (signed) measure. Now we take the Radon-Nikodym derivative of \eqref{Derivation1} with respect to the measure $\mu_{\Spb_t}$, getting
		\begin{equation*}
		\begin{split}
		P_{\sss_t} [V_t(t') \in B] &= \int_{\Omega_{X_t}}P[V_t(t')\in B|X_t = x] \mu_{X_t|\sss_t}(dx) \\
		&+ \delta(B|X_t,\sss_t)
		\end{split}
		\end{equation*}
		for some fixed $\sss_t \in \Gamma$ and where
		\begin{equation*}
		\delta(B|X_t,\sss_t) = \frac{d}{d\mu_{\Spb_t}}\Delta(B|X_t,\Gamma).
		\end{equation*}
		Using \eqref{VelProbModelB} in the last expression of $P_{\sss_t} [V_t(t') \in B]$, we obtain the claimed result. This concludes the proof.
	\end{proof}
	On the contrary to what we found for model A, where an explicit expression of $\delta(c|X_N,\sss_N)$ was given, the expression we find in general for model B is purely formal. A simple case where we may compute the Radon-Nikodym derivative $\delta(B|X_t,\sss_t)$ is described here. Note that trivially
	\begin{equation*}
	\begin{split}
	\Delta(B|X_t,\sss_t)= -\bigg[&\int_{\Omega_{X_t}} P[V_t(t')\in B|X_t = x] \mu_{X_t|\sss_t}(dx) \\
	&- P_{\sss_t}[V_t(t')\in B] \bigg].
	\end{split}
	\end{equation*}
	Assuming that all measures in the above expression admit a density with respect to the Lebesgue measure and also $\Delta(B|X_t,\Gamma)$ admits this density, i.e.
	\begin{equation*}
	\Delta(B|X_t,\Gamma) = \int_B \int_\Gamma \delta(v|X_t,\sss_t)dvd\sss_t,
	\end{equation*}
	one can immediately derive the following relation:
	\begin{widetext}
		\begin{equation}\label{DeltaModBLebesgue}
		\begin{split}
		\delta(v|X_t,\sss_t) &= -\bigg[\int_{\Omega_{X_t}} \alpha(v,x)\mu_{X_t|\sss_t}(x)dx - \mu_{V_t(t')|\sss_t}(v) \bigg] \\
		&= - \frac{1}{\mu_{\Spb_t}(\sss_t)} \bigg[\int_{\Omega_{X_t}} \alpha(v,x)\mu_{X_t,\Spb_t}(x,\sss_t)dx - \mu_{V_t(t'),\Spb_t}(v,\sss_t) \bigg].
		\end{split}
		\end{equation}
	\end{widetext}
	This expression can be considered as the analogous of equation \eqref{Velocity compl}, found for model A. In fact, as for $\delta(c|X_N,\sss_N)$, it can be computed using the joint probability densities, information that is lost after the conditioning with respect to $\Spb_t$. Thus also here, $\delta(v|X_t,\sss_t)$ contains information about the space process (as in model A).
	
	Let us now analyze the consequences of $ii)$, the following observation is useful. Consider the velocity random variable of model B. Setting $\delta t := t' - t$ we can write
	\begin{equation*}
	V_t(t') = V_t(t+\delta t) = \frac{X_{t + \delta t} - X_t}{\delta t}.
	\end{equation*} 
	Since $t$ is a parameter, we can always rescale it in order to have $\delta t = 1$. In this case, $V_t(t+1)$ resemble the velocity random variable $V_N$ of model A. This can be done for any value of $t' > t$. To make this correspondence more concrete, we may also discretize the space process. More precisely, since the points of the space process are Wiener processes taking values on $\Rea$, we can partition $\Rea$ in intervals $\{\Delta_k\}_{k \in K}$ (i.e. $\cup_{k \in K} \Delta_k = \Rea$) where $K \subset \Nat$. At this point one can consider the discretized random variable for the space process and the position random variable. If $W_t^{(i)}$ is a point of the space process $\Spb_t$, one can define a new random variable $S_t^{(i)}:\Omega_{\Spb} \rightarrow K$ as
	\begin{equation*}
	S_t^{(i)} (\omega) := k \mspace{30mu}\mbox{if}\mspace{10mu} W_t^{(i)}(\omega) \in \Delta_k,
	\end{equation*}
	which simply reveals in which $\Delta_k$ the Wiener process is. Clearly, $P[S_t^{(i)} = k] = \int_{\Delta_k} \rho_{W_t^{(i)}}(x)dx$ and given the transition probability densities for the Wiener process, say $ p^{(i)}(x,t+1; y,t)$, the transition probabilities for $S_t^{(i)}$ are given, i.e.
	\begin{equation*}
	\begin{split}
	P[    S_{t+1}^{(i)} &= k | S_t^{(i)} = j  ]\\
	&= \frac{\int_{\Delta_k}dx\int_{\Delta_j}dy p^{(i)}(x,t+1; y,t)\rho_{W_{t }^{(i)}}(y)}{\int_{\Delta_j} \rho_{W_{t }^{(i)}}(y) dy}.
	\end{split}
	\end{equation*}
	At the end of this procedure one ends up with a discretized version of the space process of model B, which is equivalent to the one used in model A. The same discretization procedure can be done for the position and velocity random variables $X_t$ and $V_{t}(t+1)$. It is not difficult to realize that theorem \ref{EURmodelA} can be applied and its application does not depend on the size of the sets $\{\Delta_k\}_{k \in K}$. Thus, as in the previous model, the requirement $ii)$ implies the entropic uncertainty relation between the position and the velocity random variables. Then as in model A, this relation can be used to prove that $X_t$ and $V_t(t')$, after conditioning on $\sss_t$, are representable as two non-commuting operators on the same Hilbert space. This will be discussed in the next section. Let us now describe a bit further how to obtain the entropic uncertainty relation from the discretization of model B. First of all, if we want to apply the results listed in \cite{LC}, we need to be sure that the two random variables are bounded, i.e. the set of all values they can assume is a bounded set. In fact, only in this case, they can be associated to two bounded self-adjoint operators, which are elements of a $C^*$-algebra, and the relation between non-commutativity and the entropic uncertainty relation holds true. In order to do that, we consider the restriction of the two random variables to a given subset. More precisely, given $\Lambda \subset \Rea = \Omega_{X_t}$, the \emph{bounded version} of $X_t$ will be the random variable
	\begin{equation*}
	X_t|_{\Lambda}(\omega) := X_t(\omega) \chi_{\Lambda} (X_t(\omega))
	\end{equation*}
	where $\chi_{\Lambda}(x)$ is the indicator function of the set $\Lambda$. Similarly, we can define the bounded version of $V_t(t')|_{\Gamma}$. At this point we consider the discrete version of these random variables, similarly to what we did for the space process. Given $X_t|_{\Lambda}$, we can discretise it simply by dividing the set $\Lambda$ in $N$ parts of equals size, obtaining a partition $\{\Delta_{N,k}^X\}_{k \in K}$, $K \subset \Nat$, such that $|\Delta_{N,k}^{X}| = |\Delta_{N,k'}^{X}|$ for any possible $k$. We can see that the number of subsets of the partition (i,e. $N = |K|$) determines the width of the sets $\Delta_{N,k}^{X}$. The bounded and \emph{discrete version} of $X_t$ is then defined as
	\begin{equation*}
	X_t^{\Delta}|_{\Lambda}(\omega) := k \mspace{30mu}\mbox{if}\mspace{10mu} X_t(\omega)|_{\Lambda} \in \Delta_{N,k}^X.
	\end{equation*} 
	The same construction can be done for the bounded version of $V_t(t')$, using in general a different partition $\{\Theta_{N',j}^V\}_{j \in J}$, obtaining $V_t(t')^{\Theta}|_{\Gamma}$. It is useful to choose the partitions for $X_t$ and $V_t(t')$ compatible with the partition used for the space process. To do that it is enough to set the partition for $X_t$ and $\Spb_t$ equal and choose the partition for $V_t(t')$ consequently. Finally we also chose to set $|\Gamma| = |\Lambda|$, i.e. the size of the two set used to bound the position and velocity random variable coincides. At this point, by discretising time as explained above $X_t^{\Delta}|_{\Lambda}$ and $V_t(t+1)^{\Theta}|_{\Gamma}$ (for simplicity we simply write $V_t^{\Theta}|_{\Gamma}$) become discrete random variables similar to those used for model A. Then, applying theorem \ref{EURmodelA}, we know that
	\begin{equation}\label{EURModelB}
	H_{\sss_t}(X_t^{\Delta}|_{\Lambda}) + H_{\sss_t}(V_t^{\Theta}|_{\Gamma}) \geqslant D,
	\end{equation} 
	where $D$ is a positive constant that in general can depend on the partition chosen but not on the probability distribution of $X_t^{\Delta}|_{\Lambda}$ and $V_t^{\Theta}|_{\Gamma}$, hence thet do not depend on $\mu_{X_t|\sss_t}$ and $\mu_{V_t|\sss_t}$. The whole construction does not depend on the partitions chosen, once they are chosen in the consistent way explained above. In particular, the above inequality holds for arbitrary partitions having small but finite size.
	
	\subsection{Construction of the Hilbert space structure for model B}\label{sec4f}
	
	The construction of the Hilbert space structure for model B goes more or less as in Model A. However, in this case, we have some additional technicalities due to the use of the partitions for the description of the two random variables involved. The entropic uncertainty relation \eqref{EURModelB}, ensures that $X_t^{\Delta}|_{\Lambda}$ and $V_t^{\Theta}|_{\Gamma}$, after conditioning on $\sss_t$, can be jointly described only on a non-commutative probability space, i.e. with non-commuting operators. Let us fix for the moment the partitions used. As in model A, the bounded and discrete version of the position random variable can be represented on the Hilbert space
	\begin{equation*}
	\Hi(X_t|N,\Lambda) = \bigoplus_{k = 1}^N \Hi_k ,
	\end{equation*}
	as the diagonal operator
	\begin{equation*}
	\hat{X}_t (N,\Lambda) = \sum_{k=1}^N k | k \rangle \langle k |.
	\end{equation*}
	Here $| k \rangle \in \Hi_{k}$ and $\hat{P}^{(\hat{X}_t (N,\Lambda))}_k := | k \rangle \langle k |$ is the PVM such that
	\begin{equation}\label{Continuospec}
	P[X_t^{\Delta}|_{\Lambda} = k] = P[X_t|_{\Lambda} \in \Delta^X_{N,k}] = \langle \psi | \hat{P}^{(\hat{X}_t (N,\Lambda))}_k| \psi \rangle
	\end{equation}
	for some $\psi \in \Hi(X_t|N,\Lambda)$. Similarly, the  bounded and discrete version of the velocity random variable can be represented on the Hilbert space
	\begin{equation*}
	\Hi(V_t|N,\Gamma) = \bigoplus_{j = 1}^N \Hi_j 
	\end{equation*}
	(note that particular partitions considered implies that $N = N'$) as the diagonal operator
	\begin{equation*}
	\hat{V}_t( N , \Gamma) = \sum_{j=1}^N j | j \rangle \langle j |.
	\end{equation*}
	The two Hilbert spaces $\Hi(X_t|N,\Lambda)$ and $\Hi(V_t|N,\Gamma)$ have the same dimension and so they are unitary equivalent, i.e. there exists a unitary map $\hat{U}:\Hi(V_t|\Gamma)\rightarrow\Hi(X_t|\Lambda)$. Hence we can represent $\hat{V}_t( N , \Gamma)$ on $\Hi(X_t|N,\Lambda)$ and viceversa. The entropic uncertainty relation \eqref{EURModelB} ensures that
	\begin{equation}\label{COMREL}
	[ \hat{X}_t (N,\Lambda) , \hat{V}_t( N , \Gamma) ] \neq 0
	\end{equation}
	Let us now analyse what happens when we change the size of the partition. First, we consider the limit $N\rightarrow\infty$ which means that the size of the partitions goes to zero. Because the sets $\Delta^X_{N,k}$ shrink to a point, say $\{x\}$, we have
	\begin{equation}\label{cruciani!!!}
	\lim_{N \rightarrow \infty} \langle \psi | \hat{P}^{(\hat{X}_t (N,\Lambda))}_k| \psi \rangle = \lim_{N \rightarrow \infty} [X_t|_{\Lambda} \in \Delta^X_{N,k}] = 0
	\end{equation}
	for any $\psi$, i.e. any $P$. This means, by prop 9.14 of \cite{moretti2013spectral}, $x \in \sigma_c(\hat{X}_t(\Lambda))$ (here $\hat{X}_t(\Lambda) := \hat{X}_t(\infty,\Lambda)$ ). By the arbitrariness of $x$ we conclude, as expected, that $\hat{X}_t(\Lambda)$ is a bounded operator with purely continuous spectrum. Note that the Hilbert space on which we can define $\hat{X}_t(\Lambda)$ is
	\begin{equation*}
	\Hi(X_t|\Lambda) := \int^{\oplus}_{\Lambda} \Hi_x dx
	\end{equation*}
	which is not separable in general. Here, $\hat{X}_t(\Lambda)$ can be written as
	\begin{equation*}
	\hat{X}_t(\Lambda) = \int_{\Lambda} x P^{(\hat{X}_t(\Lambda))}(dx).
	\end{equation*}
	Similar conclusions hold for the operator representing the bounded and discrete velocity random variable: $\hat{V}_t(\Gamma) := \hat{V}_t(\infty,\Gamma)$ is a bounded operator with continuous spectrum. Since for any value of $N$, $\hat{X}_t(N,\Lambda)$ is the operator representing the random variable obtained by discretizing the \emph{same} random variable $X_t|_{\Lambda}$, also the operators $\hat{X}_t(N,\Lambda)$ can be obtained by discretising the same operator $\hat{X}_t(\Lambda)$. The same holds for $\hat{V}_t(N,\Gamma)$. At this point because \eqref{COMREL} is valid for any possible partition chosen in the consistent way explained in the previous section (i.e. for any $N$), we can conclude that
	\begin{equation*}
	[ \hat{X}_t (\Lambda) , \hat{V}_t(\Gamma) ] \neq 0.
	\end{equation*}
	Since $\Gamma$ and $\Lambda$ are arbitrary, with similar considerations we may conclude that
	\begin{equation}
	[ \hat{X}_t , \hat{V}_t ] \neq 0
	\end{equation}
	where $\hat{X}_t$ is the unbounded operator on a Hilbert space $\Hi(X_t) := \Hi(X_t|\Rea)$ such that $\hat{X}_t (\Lambda) = \hat{P}_{\Lambda} \hat{X}_t \hat{P}_{\Lambda}$ (here $\hat{P}_\Lambda$ is the projector from $\Hi(X_t)$ to the Hilbert space $\Hi(X_t|\Lambda)$) and $\hat{V}_t$ is defined in a similar manner.
	
	We conclude by observing that $\Hi(X_t|\Lambda)$ and $\Hi(V_t|\Gamma)$ may be not separable (and so also $\Hi(X_t)$ and $\Hi(V_t)$). In general, non-separable infinite-dimensional Hilbert spaces are not mutually isomorphic. Thus in this case we cannot define a unitary map $\hat{U}:\Hi(V_t)\rightarrow\Hi(X_t)$ which maps the operator representations of $X_t$ and $V_t$ on $\Hi(V_t)$ into the corresponding operators in $\Hi(X_t)$. This is an effect of the possible lack of separability of the Hilbert spaces $\Hi(X_t)$ and $\Hi(V_t)$. However this does not mean that we cannot represent the velocity random variable on $\Hi(X_t|\Lambda)$ and vice-versa: one simply represents the velocity random variable on $\Hi(X_t|N,\Lambda)$ and then takes the limit. However, to have a consistent description the velocity operator obtained in this limit must be isomorphic to the operator $\hat{V}_t$ diagonal on $\Hi(V_t|\Gamma)$. We will refer to this problem with the name \textquotedblleft separability problem" and we will comment on it in the next section. We conclude by observing that the result obtained here, as explained in the previous section, holds for any value of $t'>t$.
	
	\subsection{Final remarks and weak points of model B}\label{modelB-final section}
	
	We completed the description of model B, which can be considered as the continuous time generalization of model A. The discreteness of time was recognized as a limitation of model A for a direct comparison to ordinary quantum mechanics. Here time is a continuous parameter as in ordinary quantum mechanics but a direct comparison is still not possible. The construction presented here leads to an infinite dimensional Hilbert space which may be not separable while in ordinary quantum mechanics the Hilbert space always is. Comparing this model with model A, we can understand that this time the number of points in the space process, $M$, does not determine the dimension of the Hilbert space. After a bit of thought, one can realize that this is a consequence of the fact that we are using probability measures which admit a density with respect to the Lebesgue measure. Another consequence of this fact is the continuous spectrum of the operators representing the particle process. However, one can always imagine that, if we let the support of the probability measure shrink to a single point (hence obtainining a Dirac measure, which is not absolutely continuous with respect to the Lebesgue measure), the operators have a pure point spectrum. This suggests that the \textquotedblleft real" continuity of the spectrum is obtained only in the limit $M = \infty$ and the absolute continuity is possible only in this case. One may observe the following. When $M \rightarrow \infty$ we can have two cases:
	\begin{enumerate}
		\item[a)] the points increase in a \emph{non dense} way: their number is infinite but in any subset of $\Rea$ these is just a finite number of them (they behaves as numbers in $\Nat$ or $\Int$);
		\item[b)] the points increase in a \emph{dense} way: their number is infinite and in any subset of $\Rea$ there is an infinite number of them (like numbers in $\mathbb{Q}$). We will refer to this case with the name \emph{dense-point limit}.
	\end{enumerate}
	Note that in both cases they are assumed to be countable. In the first case, $\hat{X}_t$ can be seen as the limit of a sequence of compact operators: the spectrum is purely point-like. However, this possibility does not seem to be comparable with the usual position operator in quantum mechanics, which is just bounded (and not compact) when we restrict it to a subset of $\Rea$. On the other hand, the second case is more interesting. Indeed, it may give rise to bounded operators which are not compact. This suggests that to completely recover quantum mechanics, the dense-point limit must be taken.
	
	Despite the observations done above, we still want to try a comparison with non-relativistic quantum mechanics. This time we are really closer to deriving the canonical commutation relation between position and momentum from the quantities of the model, as we will see. Assume the following:
	\begin{enumerate}
		\item[i)] The Hilbert space on which we can represent $\hat{X}_t$ is separable and infinite-dimensional, i.e. $L_2(\Rea)$; 
		\item[ii)] There exists a self-adjoint operator
		\begin{equation*}
		\hat{H} = \frac{1}{2m} \nabla_{x}^2 + V(x)
		\end{equation*}
		which, together with $\hat{X}_t$, fulfils all the mathematical requirements needed to apply the Ehrenfest theorem (see \cite{friesecke2009ehrenfest}).
	\end{enumerate}
	Clearly $\hat{H}$ is nothing but the ordinary hamiltonian operator in quantum mechanics. At this point, by the Ehrenfest theorem, we have the equation
	\begin{equation*}
	\frac{d}{dt} \langle \hat{X}_t \rangle_{\psi} = \frac{1}{m} \langle \hat{P}_t \rangle_\psi,
	\end{equation*}
	where $m$ is the mass of the quantum particle and $\psi \in L_2(\Rea)$. Consider now the velocity random variable of the model B
	\begin{equation*}
	V_t(t') = \frac{X_{t'} - X_t}{t' -t}.
	\end{equation*}
	Note that, after the removal of the space process, the three random variables lies in three different probability spaces and there does not exist a joint probability space where we can describe all of them (we recall that when we remove the space we use the \emph{unconditional} transition probabilities). Thus this expression is purely formal and, in particular, it is not expected to hold at the level of the outcomes of these random variables. However, the following expression makes sense
	\begin{equation*}
	\Ex[V_t(t')|\sss_t] = \frac{\Ex[X_{t'}|\sss_t] - \Ex[X_t|\sss_t]}{t' - t}
	\end{equation*}
	since the probability measures of each expectation are defined on different probability spaces. Using the procedure explained in the previous section and under the assumption $i)$, we can jointly describe these three random variable using a non-commutative probability space. In particular, we compute the expectation using the Hilbert space structure, writing
	\begin{equation*}
	\Ex[V_t(t')|\sss_t] = \langle \hat{V}_t(t') \rangle_\psi = \frac{\langle \hat{X}_{t'} \rangle_\psi -\langle  \hat{X}_t \rangle_\psi}{t' -t},
	\end{equation*}
	where $\psi \in \Hi$ with $\Hi$ Hilbert space constructed as in section \ref{sec4f}. This time an explicit procedure to construct $\psi$ is not known. From this equation we can write that
	\begin{equation*}
	\lim_{t' \rightarrow t} \langle \hat{V}_t(t') \rangle_\psi = \frac{d}{dt} \langle \hat{X}_t \rangle_{\psi} = \frac{1}{m} \langle \hat{P}_t \rangle_\psi,
	\end{equation*}
	which means that the \emph{weak}-limit $t'\rightarrow t$ of velocity operator in model B, under assumptions $i)$ and $ii)$, coincides with the momentum operator of non-relativistic quantum mechanics. Note that the assumption $i)$ on the separability of the Hilbert space is crucial for this consideration. Finally, we also note that separability also solves the problem of the non-unitary equivalence of $\Hi(X_t)$ and $\Hi(V_t)$ mentioned at the end of section \ref{sec4f}. Summarising, despite Model B is capable to reproduce the commutation relation between the position and velocity operators of the particle, which resembles the quantum mechanical commutation relation, it did not succeed in the derivation of \eqref{[Q,P]}. However, if in some other model (similar to model B) we can justify $i)$ and $ii)$ in some way, we can have a correspondence of the model with non-relativistic quantum mechanics. 
	
	\section{Conclusion}
	
	In this paper, we show how a jump-type kinematics of a point like particle together with an intrinsic stochasticity of the physical space (on which the particle moves) can give rise to a non-commutative description of the two basic observables which define the particle: its position and its velocity. In model A time is treated as a discrete parameter and this makes it impossible to have a precise comparison with the ordinary quantum mechanics. Generalizing to the continuous time case, we obtained model B. However, as pointed out in section \ref{modelB-final section}, even in this case we can not compare the two theories. Indeed, the Hilbert space on which we represent model B is non-separable in general. In both models we are able to obtain the same non-commutativity of ordinary quantum mechanics (at the algebraic level) but we have to conclude that this is not sufficient. It is worth to recall how this non-commutativity was obtained: \emph{by removing the space process at a given time $t$}. Physically this requirement is very natural: any experiment which can be done to measure the probability (via frequency), can be done \emph{in a given configuration of the space}. This means that if we assume that space really is the stochastic process described in this paper, any probability that we can measure in a laboratory is somehow conditioned to the configuration of space that we have at the time of this measurement. The fact that in our models, space is not described but removed (essentially via conditioning and not by averaging), expresses exactly this fact and is the origin of non-commutativity. We also note that in order to obtain such non-commutativity, the space process must be \emph{random}, as the entropic uncertainty relation obtained show. Indeed, if space is a deterministic phenomenon we obtain a trivial bound. The space process seems to be central, despite it must be removed to obtain a non-commutative probability space for the particle: in this sense it plays an active role in the description of the particle despite the non-commutative probability theory obtained after its removal is not capable to describe it. In addition, in model A, the space process determines the dimension of the Hilbert space, a feature which is lost in model B. This suggests that a better understanding of the space process may show a possible solution to the non-separability problem. In particular, it can be that a careful selection of a particular class of space processes, may \textquotedblleft force" the Hilbert space to be separable. This possibility will be discussed in \cite{LC3}.
	
	\section{Acknowledgements}
	
	The author would like to thank S. Bacchi, for many useful discussions and comments on the manuscript, G. Gasbarri, for its patience during the corrections and S. Marcantoni for many advices and his fundamental support during the final stage of this work. Other persons which indirectly contribute to this works are L. Bersani, N. D'Andrea, A. Motta, R. Truglia e M. Vigan\'{o}.
	
	\section{Appendix}
	
	\subsection*{A - Quantum ruler}
	
	Here we will describe an attempt to give a quantum mechanical description of a ruler, a model that will be called \emph{quantum ruler}. Let us start with the formal definition. A quantum ruler of length $N \in \Nat$ is an $N$-particle quantum system with the following features:
	\begin{enumerate}
		\item[i)] the Hilbert space of a quantum ruler is $\Hi_R := \bigotimes_{i = 1}^{N} \left( L_2(\Rea^3) \otimes \Comp^2 \right)_i$ and a generic state take this form
		\begin{equation*}
		| \psi \rangle = \psi(y_1,\cdots,y_N) \otimes | -, \cdots, - \rangle;
		\end{equation*}
		\item[ii)] there exists a region of space $L \subset \Rea^3$ such that $\langle \psi | \hat{P}^{(\hat{\mathbf{X}})}(L) |\psi \rangle = 1$, where 
		\begin{equation*}
		\hat{P}^{(\hat{\mathbf{X}})}(L) = \hat{P}^{(\hat{X}_1)}(L) \otimes \cdots \otimes \hat{P}^{(\hat{X}_N)}(L)
		\end{equation*} 
		with $\hat{P}^{(\hat{X}_i)}(\cdot)$ is the PVM associated to the position operator of the $i$-th particle;
		\item[iii)] the time evolution of a quantum ruler is determined by the hamiltonian $\hat{H}_R := \sum_{i=1}^{N} \hat{T}_i + \hat{V}$, where $\hat{T}_i$ are the $i$-th particle kinetic terms and $\hat{V} \neq \sum_{i=0}^{N}\hat{V}_i$ is some potential (chosen in order to have $\hat{H}$ bounded from below);
		\item[iv)] before any measurement the quantum ruler is described by a bounded state of the hamiltonian operator, namely $\hat{H}_R |\psi \rangle = E_0| \psi \rangle$;
		\item[v)] the measurement process of the position of a particle (call it $A$) with wave function $\phi_A(x) \in \Hi_A$ (which is not a particle of the quantum ruler) occurs with the following interaction hamiltonian (defined on the tensor product Hilbert space $\Hi_R \otimes \Hi_{Part}$)
		\begin{equation*}
		\hat{H}_I \cdot := \sum_{i=1}^{N} g \delta(y_i - x) |+_i\rangle\langle -_i|
		\end{equation*}
		where $g$ is a real constant.
	\end{enumerate}
	Let us explain the physical meaning of these requirements. A quantum ruler is a quantum system composed of $N$ distinguishable particles with spin. The spin degree of freedom should not be considered as the real spin of the particles, rather as labels that model the possibility to find the $i$-th particle in two distinguishable states. This is the content of the requirement $i)$. The condition $ii)$ simply means that, when considered as a single object, a quantum ruler is localised in a specific region $L$ of space. This assumption plays a marginal role in the rest of the analysis, nevertheless it expresses the basic fact that we cannot measure arbitrarily long distances with a given quantum ruler. Requirements $iii)$ and $iv)$ simply express that these particles are bounded together. Finally $v)$ describes how the quantum ruler measures a distance: by a spin flip. Given the particle $A$, the measurement of its position by the quantum ruler happens when \textquotedblleft $A$ touches one of the particles of the quantum ruler" causing a spin flip. After this interaction the quantum ruler undergoes a projective measurement. Since the particles are distinguishable we can label them by establishing an order, and chosing an origin (the \textquotedblleft zero" of the ruler). The position of $A$ is measured by counting the number of particles between the origin and the particle of the ruler with the spin flipped, i.e. the particle of the ruler which interacted with the particle we want to measure. If this model is correct, at least in some limit, we should be able to recover the statistics of the position of particle $A$, namely $|\phi_A(x)|^2$. Here we will show how this is possible.
	
	Let $A \subset L$ be a set. The probability to find the $i$-th particle of the ruler in this set is given by
	\begin{equation*}
	P[x_i \in A] = \langle \psi | \Id \otimes \cdots \otimes  \Id \otimes \hat{P}^{(\hat{X}_i)}(A) \otimes \Id \otimes \cdots \otimes \Id | \psi \rangle
	\end{equation*}
	where $| \psi \rangle$ is the state of the quantum ruler. From ordinary quantum mechanics, we know that
	\begin{equation*}
	P[x_i \in A] =  \int_A \rho_i(x) dx
	\end{equation*}
	where $\rho_i(x)$ is a density with respect to the Lebesgue measure. An example of possible probability densities of a quantum ruler is drawn in figure \ref{fig:1}-$a)$.
	\begin{figure*}[h!t!]
		\includegraphics[scale=0.5,angle=270]{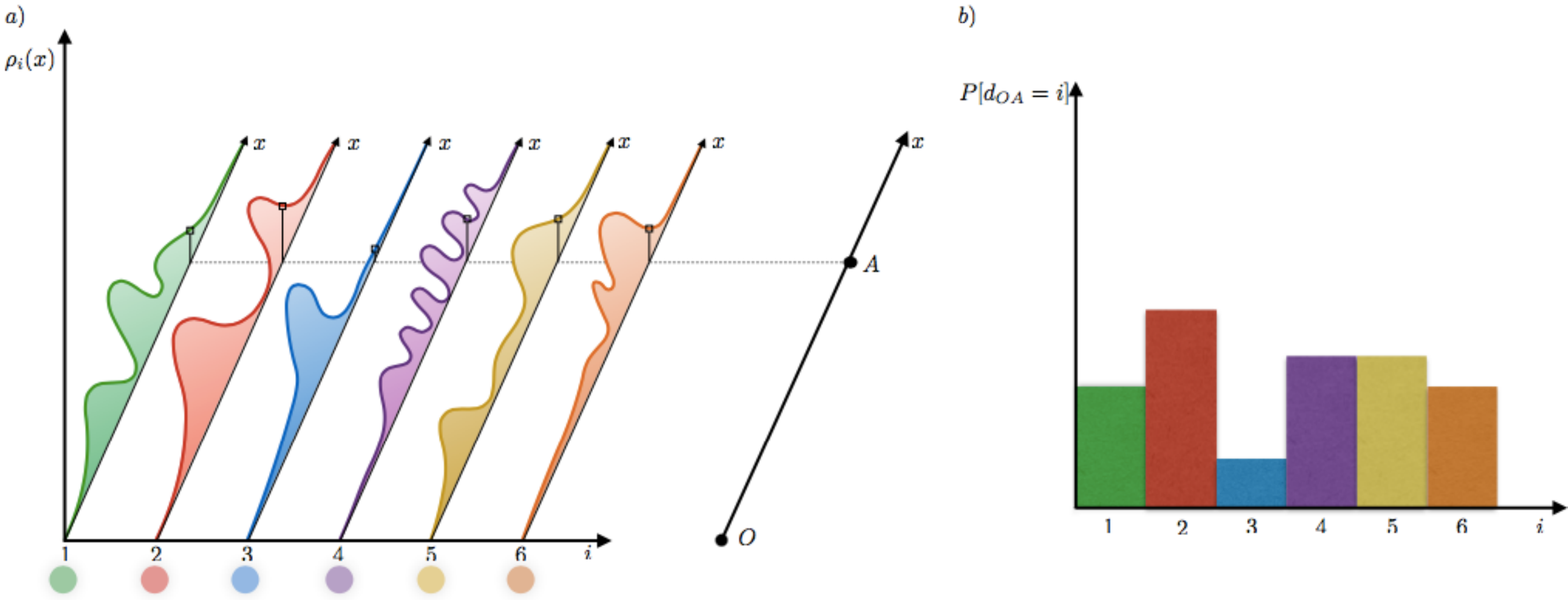}
		\caption{In figure $a)$ the $\rho_i(x)$ of a $N=7$ quantum ruler are drawn: the $\rho_i$ of all particles of the ruler is drawn except the origin. The origin is assumed to be a particle in the point $O$, i.e. $\rho_O(x) = \delta(x)$ (and it is not drawn), for simplicity. In figure $b)$ the probability distribution for the position of the point $A$ (obtained by the same quantum ruler. The probability $P[d_{AO} = i]$ drawn should be considered as the probability distribution of the point $A$ (na\"ively, a particle with a delta-like wave function centred in $x_A$) measured with the quantum ruler with origin chosen to be the green particle.}
		\label{fig:1}
	\end{figure*}
	From the figure, it is easy to understand that the quantum ruler can give us only statistical information about a possible measurement of the point $A$. Such probability distribution is showed in figure \ref{fig:1}-$b)$. Note that we did not use the requirement $v)$ to arrive at this conclusion,  hence we are not really describing the measurement of the position of the particle. In order to do that precisely, let us assume that we have a particle $A$, described by a wave function $\phi_A(x) \in \Hi_A$. The whole system (quantum ruler plus particle) will be described by the state
	\begin{equation*}
	|\psi_t \rangle = \psi_R(y_1,\cdots,y_N)\otimes | -, \cdots, - \rangle \otimes \phi_A(x)
	\end{equation*}
	which is  an element of $\Hi_R \otimes \Hi_A$. The measurement process is modelled by the interaction hamiltonian as in point $v)$. After the interaction we perform a projective measurement to learn which spin is flipped: the number of particles between the origin (which is an arbitrarily chosen particle of the quantum ruler) and the particle with the spin flipped is our distance. Hence, if $d(O,A)$ is the distance between the origin and the particle with the spin flipped, we can write that
	\begin{equation*}
	P[d(O,A) = i] = P[S_i = +_i] = \langle \psi_{s} | (|+_i \rangle \langle +_i |)| \psi_{s} \rangle
	\end{equation*}
	where $|\psi_{s} \rangle$ is the total quantum state right after the interaction. We will compute $|\psi_{s} \rangle$ in interaction picture for $s-t =\delta \tau$ such that $g \delta \tau \ll 1$. In this regime we have that
	\begin{equation*}
	e^{-i\delta \tau \hat{H}_I} \approx \Id - ig\delta\tau\sum_{i=1}^{N}\delta(y_i - x) |+_i\rangle\langle -_i| + O(g^2\delta\tau^2)
	\end{equation*}
	Thus, setting $t = 0$
	\begin{equation*}
	\begin{split}
	| \psi_{\delta \tau} \rangle  \approx | \psi_0 \rangle - ig\delta\tau &\sum_{i=1}^{N}\delta(y_i - x) \psi_R(y_1,\cdots,y_i,\cdots,y_N) \otimes\\
	&\otimes | -, \cdots, +_i ,\cdots, - \rangle \otimes \phi_A(x) + O(g^2\delta\tau^2)
	\end{split}
	\end{equation*}
	Using this result, we obtain
	\begin{equation*}
	\begin{split}
	P[ & d(O,A) = i] = \langle \psi_{\delta \tau} | (|+_i \rangle \langle +_i |)| \psi_{\delta \tau} \rangle \\
	&\approx g^2 \delta \tau^2\int \prod_{j = 1}^N dy_j dx \delta^2 (y_i - x) |\psi_R(y_1,\cdots, y_N)|^2 |\phi_A(x)|^2 \\
	&= g^2 \delta\tau^2 \delta(0)\int \prod_{j = 1}^N dy_j  |\psi_R(y_1,\cdots,y_{i},\cdots, y_N)|^2 |\phi_A(y_i)|^2.
	\end{split}
	\end{equation*}
	The $\delta(0)$ is due to the presence of the square of the Dirac-delta in the probability density computed using the wave function evolved in the interaction picture. In order to deal with this divergent term, we require that $g^2\delta\tau^2\delta(0) \approx 1$. This requirement is in agreement with the fact that the terms $g^2\delta\tau^2$ are infinitesimal and they can be neglected in the expression of $|\psi_{\delta\tau}\rangle$. Thus we conclude that
	\begin{equation}\label{QuantumRulerProb}
	\begin{split}
	P[ & d(O,A) = i] =\\
	&= \int \prod_{j = 1}^N dy_j  |\psi_R(y_1,\cdots,y_{i},\cdots, y_N)|^2 |\phi_A(y_i)|^2
	\end{split}
	\end{equation}
	As we can see, the probability of the outcome depends both on the ruler and particle state. We do not obtain $|\phi_A(y_i)|^2$, because we are not using a classical ruler. In order to obtain this result we may take the \textquotedblleft dense limit", which we interpret as the following
	\begin{equation*}
	N\rightarrow \infty, \mspace{30mu} \eta \rightarrow \infty, \mspace{30mu} L = \mbox{costant}
	\end{equation*}
	where $\eta$ is the particle density, i.e. $\eta = N / L$. Since the ratio $N/\eta$ must remain constant, the density must go to infinity as $N$ and in the same way in any point of the volume. To realise this situation, we may imagine that as $N$ increases, the particles of the quantum ruler are described by gaussian wave functions centred around different points of space. In order to keep $\eta / N$ constant, as $N$ increases the overlaps between the gaussian should reduce. This means that the square modulus of the wave function tends to a Dirac delta. Hence, in the \textquotedblleft dense limit" we can formally write that
	\begin{equation*}
	|\psi_R(y_1, \cdots, y_N)|^2  \rightarrow \prod_{j = 1}^N \delta(y_j - X_j) \mbox{ for } N \gg 1,
	\end{equation*}
	where $X_j \in L$ are the points of the quantum ruler. Substituting this expression in \eqref{QuantumRulerProb}, we obtain $P[d(O,A) = i] = |\phi_A(X_i)|^2$. If in the \textquotedblleft dense limit" considered, we label the points of the ruler with its coordinate with respect to a chosen origin, we can suppress the index $i$ obtaining $P[d(O,A) = X] = |\phi_A(X)|^2$. Hence we can see that the quantum ruler reduces to a classical ruler (thought as a solid continuous rod) in the \textquotedblleft dense limit" described above.\newline
	
	\paragraph*{Remark.} Note that the quantum ruler is \textquotedblleft quantum", only because the probability distributions used are derived according to the quantum formalism. It is not difficult to see that if we replace the probability distributions arising from $\psi_R(y_1, \cdots, y_N)$ with a probability distribution arising from $M$, suitably correlated, stochastic processes the whole description of the (\emph{stochastic}, this time) ruler would be the same.
	
	\subsection*{B - Distance between two points $A$ and $B$ of a random distribution of points in a set $\Lambda$}
	
	Let the symbol $\mathcal{X}(\Lambda)$ label a random distribution of points over the set $\Lambda$. In this appendix, we will describe two possible methods to introduce a notion of distance between two points belonging to  $\mathcal{X}(\Lambda)$. In both methods presented here, we will try to define the distance between two points $A$ and $B$ using only the other points of $\mathcal{X}(\Lambda)$, which we hope will clarify the expression \textquotedblleft measured on the points " used for example in section \ref{ModA:SPsec}. 
	
	Before introducing the two aforementioned distances, let us define what we mean with the term distance in this appendix.
	\begin{definition}\label{semi-defi}
		Let $G$ be a set and $d: G \times G \rightarrow \Rea^+$ be a function. Given two points $x,y \in G$, we say that $d(x,y)$ is the \emph{distance between $x$ and $y$} if
		\begin{enumerate}
			\item[i)] $d(x,y) = d(y,x)$ for all $x,y \in G$;
			\item[ii)] for any $x,y \in G$, $d(x,y) = 0$ if and only if $x = y$;
		\end{enumerate}
	\end{definition}
	Another name used in the literature for the function $d$ of the above definition is \emph{semi-metric} and the couple $(G,d)$ is called \emph{semi-metric space}. 
	\begin{definition}
		Let $G$ be a set and $d : G \times G \rightarrow \Rea^+$ be a distance on it. If for any $x,y,z \in G$
		\begin{equation*}
		d(x,y) + d(y,z) \geqslant d(x,z),
		\end{equation*}
		the distance is said \emph{metric}.
	\end{definition}
	As in the previous case, the couple $(G,d)$ has a particular name: \emph{metric space}. Semi-metric spaces and metric spaces are related by the following interesting theorem \cite{wilson1931semi}.
	\begin{theorem}
		Let $(G,d)$ be a semi-metric space. Assume that, for any  $a \in G$ and any $k \in \Rea^+$ there exists $r \in \Rea^+$ such that, given any point $b \in G$ for which $d(a,b) \geqslant k$, then
		\begin{equation*}
		d(a,c) + d(b,c) \geqslant r
		\end{equation*}
		holds for any $c \in G$. Then $(G,d)$ is homeomorphic to a metric space.        
	\end{theorem}
	Hence if the conditions of this theorem are fulfilled, then there exist a continuous function between the original semi-metric and some metric space that has a continuous inverse function (i.e. an homeomorphism). Loosely speaking, the distance in a the semi-metric space can be \textquotedblleft distorted" into a distance over a metric space.
	
	The first distance on $\mathcal{X}(\Lambda)$ we introduce will be called \emph{nearest neighbourhood distance}, or \emph{NNG-distance} for short. Let us assume that $(\Lambda,h)$ is a metric space. Let $A$, $B$ and $C$ be points in $\mathcal{X}(\Lambda)$, thus they are also points in $\Lambda$. In order to define the NNG distance between two points in $\mathcal{X}(\Lambda)$ we need to introduce a \emph{selection procedure} which we will use to understand which is, given a point $A$, its closest point in $\mathcal{X}(\Lambda)$ (see figure \ref{fig:NNGsp}). This selection procedure makes explicitly use of the underlying metric structure of  $\Lambda$. In particular, if $A \in \mathcal{X}(\Lambda)$, its \emph{closest point} is the point $B \in \mathcal{X}(\Lambda)$ which minimises the distance $h(x,A)$ for all $x \in \mathcal{X}(\Lambda)/\{A\}$. In symbols, the closest point to $A$ in $\mathcal{X}(\Lambda)$ can be defined as
	\begin{equation*}
	\mbox{Cl}(A|\mathcal{X}(\Lambda)) := \{x \in \mathcal{X}(\Lambda) \mbox{ }|\mbox{ } \min_{y \in \mathcal{X}(\Lambda)/\{A\} } h(y,A) = h(x,A) \}  .
	\end{equation*}
	\begin{figure}[h!]
		\includegraphics[scale=0.35]{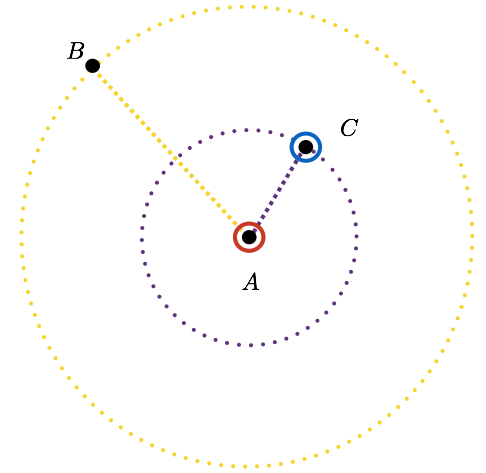}
		\caption{The selection procedure for the NNG-distance in simple the case of $3$ points. In this case $C = \mbox{Cl}(A|\mathcal{X}(\Lambda))$ where $\mathcal{X}(\Lambda) = \{A,B,C\}$. Note that it depends on the underlying metric structure (euclidean in this case).}
		\label{fig:NNGsp}
	\end{figure} 
	The idea behind the NNG-distance of two points $A$ and $B$ is to count the number of points that we need to find to arrive in $B$, excluding all the previous closest points we found. Let us explain better this idea. Starting from $A$, the closet point is $x_1= \mbox{Cl}(A|\mathcal{X}(\Lambda))$. Clearly if $x_1$ is the closest point of $A$, then also the converse is true, i.e. $A = \mbox{Cl}(x_1|\mathcal{X}(\Lambda))$, and clearly we never reach the point $B$ by iterating this procedure. The simplest way out is to look for the closest point to $x_1$ \emph{excluding} $A$, i.e. $x_2 = \mbox{Cl}(x_1|\mathcal{X}(\Lambda)/\{A\})$. By repeating this procedure till we do reach point $B$, we select a collection of points $D(A,B):=\{x_1,\cdots,x_M = B\}$ (the generic point of this collection is $x_i = \mbox{Cl}(x_{i-1}|\mathcal{X}(\Lambda))/\{x_{i-2},\cdots,x_1\}$) and we call $\delta(A,B)$ the number of points in this collection, i.e. $\delta(A,B) = |D(A,B)|$. However the function $\delta$ is not symmetric under the exchange of its arguments in general (see figure \ref{fig:NNGd} for an example).
	\begin{figure}[h!]
		\includegraphics[scale=0.35]{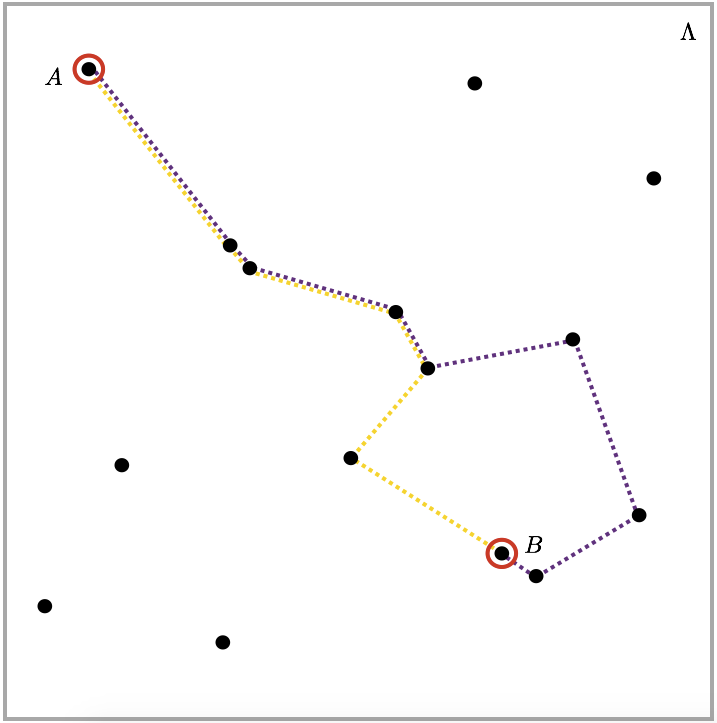}
		\caption{An example of NNG-distance. In this case $d_1(A,B) = 7$. Note that $\delta(A,B) = 6$ (yellow line) and $\delta(B,A) = 8$ (violet line), showing that $\delta$ is not symmetric in general.}
		\label{fig:NNGd}
	\end{figure} 
	However, it is a known fact that any non-symmetric function can be always be symmetrised. We define the NNG-distance as a symmetrised version of $\delta$.
	\begin{definition}
		Let $A,B \in \mathcal{X}(\Lambda)$ be two points. The \emph{NNG-distance} between $A$ and $B$ is defined as
		\begin{equation*}
		d_1(A,B) = \frac{\delta(A,B) + \delta(B,A)}{2}.
		\end{equation*}        
	\end{definition}
	In general, $d_1$ is not a metric distance. Note that in the above explanation for the construction of this distance, we did not consider the case of possible ambiguities in the selection procedure, namely the possibility to have two points with the same distance. In this case, one may go on with the selection procedure using one of the two points for each ambiguity, which means constructing two (different in general) collections of points $D(A,B)$, and define the NNG-distance using the inferior of the two $\delta(A,B)$ obtained with respect to the two collections. To conclude the discussion about the NNG-distance, we give some physical motivation regarding this definition. First, the collection of points $\{x_1,\cdots,x_M\}$ can be seen as the number of particles of a (stochastic or quantum) ruler measuring the distance between $A$ and $B$. Given that, $d_1(A,B)$ can be seen as the distance covered by a particle, which can jump from one point to its closest in a fixed amount of time (i.e. with constant speed), from the point $A$ to reach the point $B$ and then come back in $A$. This resembles the radar method used in special and general relativity to define distances.
	
	Let us now introduce the second distance on $\mathcal{X}(\Lambda)$ which will be called \emph{triangular distance}, or \emph{T-distance} for short. We have seen that the NNG-distance strongly depends on the underlying metric structure of $\Lambda$. The  T-distance is an attempt to reduce this dependence. The idea is schematically explained in figure \ref{fig:Td}.
	\begin{figure}[h!]
		\includegraphics[scale=0.35]{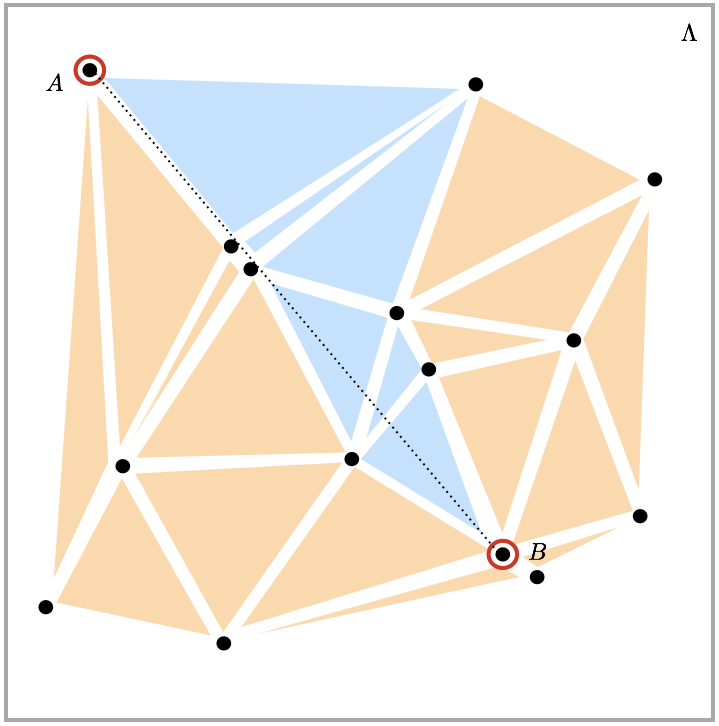}
		\caption{An example of T-distance. In this case $d_2(A,B) = 6$. Note that the T-distance is symmetric. The collection of points $\mathcal{X}(\Lambda)$, $A$ and $B$ are the same used in figure \ref{fig:NNGd}.}
		\label{fig:Td}
	\end{figure}
	The T-distance can be obtained as follows. Given the distribution of point $\mathcal{X}(\Lambda)$, construct all the triangles, whose vertices are the points in $\mathcal{X}(\Lambda)$ which do not have any point of $\mathcal{X}(\Lambda)$ inside them. In general if $\Lambda $ is a $d$-dimensional space, one constructs all the $d$-dimensional generalisation of a triangle, namely a \emph{$(d-1)$-simplex}. If there is an ambiguity, i.e. from a set of points one can draw equivalently two couples of triangles (we are considering the $2$-D case), draw first the triangle with the smallest area, computed via \emph{Pick's theorem} \cite{pick1899geometrisches} to avoid the use of the underlying space $\Lambda$. In the $d$-dimensional case, instead of using the area, one considers the $d$-volume, which can be computed in a background independent way using the Ehrhart polynomial \cite{ehrhart1962geometrie}. Then the T-distance is defined as
	\begin{definition}
		Given two points $A,B \in \mathcal{X}(\Lambda)$, the \emph{T-distance} between $A$ and $B$ is defined as
		\begin{equation*}
		d_2(A,B) = \{ \mbox{Number of triangles touched by the line $\overline{AB}$}\}.
		\end{equation*}
	\end{definition}
	Note that $d_2(A,B)$ is automatically symmetric, hence it is a semi-metric. However the triangular inequality does not hold in general, as figure \ref{fig:Tdvio} shows.
	\begin{figure}[h!]
		\includegraphics[scale=0.35]{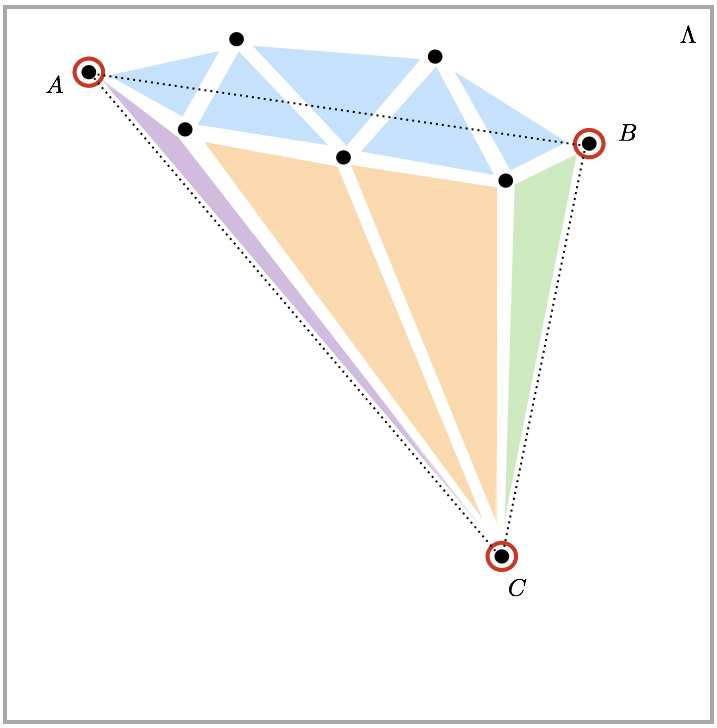}
		\caption{An example where the T-distance violates the triangular inequality: $d_2(A,C) + d_2(C,B) = 2 < d_2(A,B) = 5$.}
		\label{fig:Tdvio}
	\end{figure}
	From the physical point of view, this distance can be interpreted as the number of particles needed to a quantum (stochastic) ruler to measure the distance between $A$ and $B$ \textquotedblleft minimizing the $d$-volume occupied" by the reler. Indeed, the number of triangles corresponds to the number of points between $A$ and $B$ which are the vertices of the triangles as well. Other distances which are \textquotedblleft more background independent" may be available (for example one can define a distance as the smallest number of triangles constructed as before, which link, a triangle having $A$ as vertex, with another triangle having $B$ as vertex), but the discussion of them is out of the scope of this appendix.
	
	\paragraph*{Remark.} The distances presented here can be applied in the 1-D case. The NNG-distance can be applied without problems while for the $T$-distance we recall that a $0$-simplex is simply a point (thus $d_2(A,B)$ is just the numbers of points between $A$ and $B$).
	
	\subsection*{C - Measure-theoretic conditional probability and conditional expectation}\label{Appendix: measure theoretic conditional probability}
	
	In this appendix, a concise explanation of the notion of conditional expectation (and so conditional probability) in a measure-theoretic setting is given. The discussion done here is based on Sec. 2.3 of \cite{gudder2014stochastic} and on Ch. 5 of \cite{kallenberg2006foundations}, which are the main references for the interested reader. A more coincised exposition of the topic can be found in Sec. 2.7 of \cite{klebaner2005introduction}.
	
	In order to introduce the notion of conditional expectation starting from elementary probability, consider the following elementary definition.
	\begin{definition}
		Let us consider a probability space $(\Omega,\E,P)$ and let $B \in \E$ be such that $P(E) > 0$. For any $A \in \E$, the \emph{conditional probability of $A$ given $B$} is given by
		\begin{equation*}
		P(A|B) := \frac{P(A \cap B)}{P(B)}.
		\end{equation*}
	\end{definition}
	It is not difficult to see that, for a given $B \in \E$, the map $P(\cdot | B): A \mapsto P(A|B)$ is again a probability measure: this fact allows to interpret $P(A | B)$ as the probability that the event $A$ occurs given the fact that the event $B$ has already occurred. Now, suppose we have a collection of sets $\{B_i\}_{i \in I}$, elements of $\E$, which are mutually disjoint $B_i \cap B_j = \{\varnothing\}$ for any $i,j \in I$. The smallest $\sigma$-algebra, say $\E_0$,  which contains all these sets is called \emph{$\sigma$-algebra generated by $\{B_i\}_{i \in I}$}. In symbols, we will write $\E_0 = \sigma( \{B_i\}_{i \in I}) $. We can generalise the notion of conditional expectation given a set (i.e. an event) in the following way.
	\begin{definition}
		Let $(\Omega,\E,P)$ be a probability space and let $\{B_i\}_{i \in I}$ be a family of mutually disjoint subsets of $\Omega$. Assume that for any $i \in I$, $P(B_i) >0$. Let $\E_0 = \sigma (\{B_i\}_{i \in I})$. For any $A \in \E$, the \emph{conditional probability of $A$ given $\E_0$} is defined to be the random variable
		\begin{equation*}
		P(A|\E_0) (\omega) = P(A|B_i)
		\end{equation*}
		for $\omega \in B_i \in \E_0$.
	\end{definition}
	We want to stress that the conditional probability with respect to a $\sigma$-algebra is a random variable and not a number, as in the initial case of conditional probability with respect to an event. The conditional probability with respect to a given $\sigma$-algebra has two important properties. First, by construction $P(\cdot | \E_0)$ is measurable with respect to $\E_0$. In addition we also have the following proposition.
	\begin{proposition}
		Let $(\Omega,\E,P)$ be a probability space and take a family of mutually disjoint sets $\{B_i\}_{i \in I}$ all belonging to $\E$. Let $P(\cdot|\E_0) (\omega)$ be the conditional expectation given $\E_0 = \sigma(\{B_i\}_{i \in I})$, then
		\begin{equation*}
		\int_{B}P(A|\E_0) (\omega)P(d\omega) = \int_{B} \chi_A(\omega) P(d\omega)
		\end{equation*}
		for all $B \in \E_0$.
	\end{proposition}
	\begin{proof}
		Assume $B = B_i$ for some $B_i$ of the family of  set generating the $\sigma$-algebra $\E_0$. In this case we can write
		\begin{equation*}
		\begin{split}
		\int_{B_i} P(A|\E_0)(\omega) P(d\omega) &= P(A|B_i) \int_{B_i} P(d\omega) \\
		&= \frac{P(A \cap B_i)}{P(B_i)} P(B_i) \\
		&= \int_{\Omega} \chi_{A \cap B_i}(\omega)P(d\omega) \\
		&= \int_{B_i} \chi_A(\omega) P(d\omega).
		\end{split}
		\end{equation*}
		For a generic $B \in \E_0$ the result holds by additivity of the integral. Indeed a generic element of $\E_0$, if it is not an element of the family, is always the union of two or more elements of the family. This concludes the proof.
	\end{proof}
	The reason why these two properties of $P(\cdot|\E_0)$ are interesting, is because they completely characterise the conditional expectation in a measure-theoretic sense, as the following proposition shows.
	\begin{proposition}\label{PropCondProb}
		Let $(\Omega,\E,P)$ be a probability space and take a family of mutually disjoint sets $\{B_i\}_{i \in I}$ all belonging to $\E$. Let $\E_0 = \sigma(\{B_i\}_{i \in I})$ and $A \in \E$ an event. Assume that there exists a random variable $f_A:\Omega \rightarrow \Rea$ on this probability space such that:
		\begin{enumerate}
			\item[i)] $f$ is measurable with respect to $\E_0$;
			\item[ii)] for any $B \in \E_0$, the following
			\begin{equation*}
			\int_{B} f_A P(d\omega) = \int_B \chi_A P(d\omega)
			\end{equation*}
			holds.
		\end{enumerate}
		Then $f_A(\omega) = P( A | \E_0)(\omega)$.
	\end{proposition}
	\begin{proof}
		Because $\E_0$ is generated by $\{B_i\}_{i \in I}$ and $f_A$ is measurable with respect to it, then $f$ is constant on each $B_i$. Then, for every $\omega \in B_i$, the second requirement implies
		\begin{equation*}
		\begin{split}
		f_A(\omega) &= f_A(\omega) \frac{\int_{B_i} P(d\omega)}{\int_{B_i}P(d\omega)} \\
		&= \frac{1}{P(B_i)} \int_{B_i} f_A(\omega) P(d\omega)\\
		&= \frac{1}{P(B_i)} \int_{B_i} \chi_A(\omega) P(d\omega)\\
		&= \frac{1}{P(B_i)} \int_{B_i \cap A}P(d\omega)= P(A|B_i)
		\end{split}
		\end{equation*}
		which concludes the proof.
	\end{proof}
	At this point, it is quite natural to define the conditional expectation with respect to an event, as the ordinary expectation value with respect to the conditional probability with respect to this event. In particular, assume that $B \in \E$, then we can define
	\begin{equation*}
	\begin{split}
	\Ex[ f | B ] &= \int_\Omega f(\omega) P(d\omega|B) \\
	&= \int_\Omega f(\omega) \frac{P(d\omega \cap B )}{P(B)} \\
	&= \frac{1}{P(B)} \int_B f(\omega) P(d\omega).
	\end{split}
	\end{equation*}
	Then, we can generalise this conditional expectation to the case of a $\sigma$-algebra generated by a family $\{B_i\}_{i \in I}$. If $\E_0 = \sigma( \{B_i\}_{i \in I} )$, using the definition stated above for $\Ex[ f | B ]$, we define the random variable
	\begin{equation*}
	\Ex[ f | \E_0 ](\omega) = \Ex[ f | B_i ] 
	\end{equation*}
	for $\omega \in B_i$. By the proposition \ref{PropCondProb}, one can easily conclude that
	\begin{enumerate}
		\item[i)] $\Ex[ f | \E_0 ]$ is measurable with respect to $\E_0$;
		\item[ii)] for any $B \in \E_0$, then
		\begin{equation*}
		\int_{B} \Ex[ f |\E_0] (\omega) P(d\omega) = \int_B f(\omega) P(d\omega).
		\end{equation*}
	\end{enumerate}
	These two properties completely characterise the conditional expectation with respect to $\E_0$. The arguments presented till here should justify the following definition.
	\begin{definition}\label{CondExpDef}
		Let $(\Omega,\E,P)$ be a probability space and let $\mathcal{F}$ be a  sub-$\sigma$-algebra of $\E$, namely $\mathcal{F} \subset \E$. Consider a random variable $X:\Omega \rightarrow \Rea$, which is integrable, i.e. $X \in L_1(\Omega,P)$. The \emph{conditional expectation of $X$ with respect to $\mathcal{F}$}, $\Ex[ X | \mathcal{F}]$ is the random variable such that
		\begin{enumerate}
			\item[i)] $\Ex[ X | \mathcal{F} ]$ is measurable with respect to $\mathcal{F}$;
			\item[ii)] for any $B \in \mathcal{F}$, then
			\begin{equation*}
			\int_B \Ex[ X | \mathcal{F} ] P(d\omega) = \int_B X P(d\omega).
			\end{equation*}
		\end{enumerate}
	\end{definition}
	Note that $\mathcal{F}$ is a generic sub-$\sigma$-algebra without any reference to a family of sets. The existence of $\Ex[ X | \mathcal{F} ]$ is ensured by the Radon-Nikodym theorem and so it is unique up to $P$-null sets. Given the conditional expectation, the conditional probability can be simply defined as the conditional expectation of the characteristic function of an event, namely
	\begin{equation*}
	P(A|\mathcal{F}) := \Ex [\chi_A | \mathcal{F}].
	\end{equation*}
	It is not difficult to see that the particular cases presented in the beginning to justify the definition \ref{CondExpDef} are contained in $\Ex[ X | \mathcal{F} ]$. This general definition allows to define the conditional expectation of a random variable with respect to the other. 
	\begin{definition}
		Given a probability space $(\Omega,\E,P)$  a measurable space $(F,\F)$ and a random variable $Y:\Omega \rightarrow F$, consider the family of set $\{ Y^{-1}(B) \}_{B \in \F}$, and call $\sigma(Y):= \sigma(\{ Y^{-1}(B) \}_{B \in \F})$. The conditional expectation of another random variable $X$ with respect to $Y$, $\Ex[X|Y]$, is defined as
		\begin{equation*}
		\Ex[X|Y] := \Ex[X | \sigma(Y)].
		\end{equation*}
	\end{definition}
	From this conditional expectation we can clearly obtain the conditional probability $P_Y(A) = P( A |Y)$ setting $X = \chi_A$, where $A$ is an event.    In the next proposition we will list without proof some of the basic properties of the conditional expectation.
	\begin{proposition}\label{PropCondProb1}
		Let $(\Omega,\E,P)$ be a probability space, $X,Y \in L_1(\Omega,P)$ integrable random variables on it and $\mathcal{F},\mathcal{G} \subset \E$ sub-$\sigma$-algebras. Then, up to $P$-null sets:
		\begin{enumerate}
			\item[i)] $\Ex[XY|\mathcal{F}] = X \Ex[Y| \mathcal{F}]$, when $X$ is $\mathcal{F}$-integrable;
			\item[ii)] $\Ex[\Ex[X|\mathcal{F}]] = \Ex[X]$, which is called \emph{law of total expectation};
			\item[iii)] $\Ex[ \Ex[X|\mathcal{F}] | \mathcal{G} ] = \Ex[ X | \mathcal{G} ]$, when $\mathcal{F} \subset \mathcal{G}$.
		\end{enumerate}
	\end{proposition}
	Note that $ii)$ is just a particular case of $iii)$ when $\mathcal{G}$ is the trivial $\sigma$-algebra. Finally, to conclude this appendix, we want to derive the usual formula for the conditional probability density, i.e.
	\begin{equation}\label{CondProbForm}
	\rho_{X|Y = y}(x) = \frac{\rho_{X,Y}(x,y)}{\rho_Y(y)},
	\end{equation}
	starting from the given measure-theoretic definition of conditional expectation. In order to do that, we need to introduce the notion of \emph{regular conditional probability}.
	\begin{definition}
		Let $(\Omega,\E,P)$ be a probability space, $(F,\F)$ a measurable space and $Y:\Omega \rightarrow F$ a random variable. A \emph{probability kernel} is a function $\rho: F \times \E \rightarrow [0,1]$ such that
		\begin{enumerate}
			\item[i)] the map $y \mapsto \rho(y,A)$ is a measurable function on $(F,\F)$, for any fixed $A \in \E$;
			\item[ii)] for any $A \in \E $, the map $A \mapsto \rho(y,A)$ is a probability measure on $(\Omega,\E)$, for any fixed $y \in F$.
		\end{enumerate}
		A probability kernel is said to be a \emph{regular conditional probability} if in addition
		\begin{equation}\label{RegProb}
		P(A \cap Y^{-1}(B)) = \int_B \rho(y,A) (P\circ Y^{-1})(dy).
		\end{equation}
	\end{definition} 
	Probability spaces, where all conditional probabilities are regular for any random variable, are said to have the \emph{regular conditional probability property}. Let us explain the definition above to have a better understanding of the meaning of regular conditional probability. We recall that $P \circ Y^{-1}$ is nothing but the image of the probability measure on $(F,\F)$ via $Y$, namely the probability distribution of $Y$. Thus we can rewrite \eqref{RegProb} as
	\begin{equation*}
	P(A \cap Y^{-1}(B)) = \int_B \rho(y,A) \mu_Y(dy).
	\end{equation*}
	Written in this way, this formula can be seen as the \textquotedblleft continuous version" of the Bayes formula (more generally the Bayes formula can be seen as the discrete case of the law of total expectation seen in the proposition \ref{PropCondProb1}). This suggests the following: $\rho(y,A) = P( A | Y = y)$. In order to make this statement rigorous we need some regularity condition on the measurable space $(F,\F)$.
	\begin{theorem}\label{TeoAppC-WellDef}
		Let $(\Omega,\E,P)$ be a probability space and let $(F,\F)$ and $(G,\mathcal{G})$ be two measurable spaces. Assume that $(G,\mathcal{G})$ is also a Borel space. Consider two random variables $Y: \Omega \rightarrow F$ and $X: \Omega \rightarrow G$. Then there exists a probability kernel $\rho: G \times \F \rightarrow [0,1]$ such that
		\begin{equation*}
		P(X \in A| Y = y) = \rho(y , A) \qquad P - a.s.
		\end{equation*}
		for $A \in \mathcal{G}$, which means that it is a regular conditional probability. $\rho$ is unique up to $ P\circ Y^{-1}$-null sets.
	\end{theorem}
	At this point, the disintegration theorem \cite{kallenberg2006foundations} ensures that we can compute the conditional expectation using the measure $\rho(y,A)$.
	\begin{theorem}\label{DisTheo}
		Let $(\Omega,\E,P)$ be a probability space, $(F,\F)$ and $(G,\mathcal{G})$ be two measurable spaces and $X: \Omega \rightarrow G$ a random variable. Assume that $(F,\F)$ and $(G,\mathcal{G})$ are such that the probability kernel $\rho: G \times \F \rightarrow [0,1]$ is a regular conditional probability. Consider a random variable $Y: \Omega \rightarrow G$ and set $\mathcal{T} = \sigma(Y) \subset \mathcal{G}$. If $ Z: \Omega \rightarrow F \times G$ is a $\E$-measurable function such that $\Ex[ |Z(X,Y)| ]< + \infty$, then
		\begin{equation*}
		\begin{split}
		\Ex[Z(X,Y) | \mathcal{T}] &= \Ex[Z(X,Y)|Y] \\
		&= \int_F Z(x,Y) \rho(Y,dx) \qquad P-a.s.
		\end{split}
		\end{equation*}
	\end{theorem}
	This theorem ensures that, under suitable conditions, we can use a regular probability measure to compute the conditional expectation, as the previous discussion suggested. Now we are ready to derive the conditional probability density formula \eqref{CondProbForm}. As always, given a probability space $\PS$ and two measurable spaces $(F,\F)$ and $(G,\mathcal{G})$ on which the random variables $X$ and $Y$ take values , respectively. Assume that we can define a regular probability measure $\rho(x,A)$ and that:
	\begin{enumerate}
		\item[i)] $\rho(y,A)$ can be written as $\rho(y,A) = \int_A \mu(y,x)dx$;
		\item[ii)] $[P \circ Y^{-1}](B)$ admits density $\rho_Y(y)$ with respect to the Lebesgue measure;
		\item[iii)] $[P \circ (X^{-1},Y^{-1})] (C,B)$ admits density $\rho_{X,Y}(x,y)$ with respect to the Lebesgue measure.
	\end{enumerate}
	Then, theorem \ref{DisTheo} ensures that we can compute the conditional expectation using $\rho$. Choose $A$ as the event $X^{-1}(C) :=\{X \in C\}$. By definition we can write
	\begin{equation*}
	P(X^{-1}(C) \cap Y^{-1}(B)) = \int_B \int_C \mu(y,x) \rho_Y(y) dxdy
	\end{equation*}
	on the other hand, if we choose $A$ as the event $\{X \in C\}$, we can write
	\begin{equation*}
	P(X^{-1}(C) \cap Y^{-1}(B)) = \int_{C \times B} \rho_{X,Y}(x,y) dxdy.
	\end{equation*}
	Therefore, we may conclude that up to $dxdy$-null sets
	\begin{equation*}
	\rho_{X,Y}(x,y) = \mu(x,y)\rho_{Y}(y)
	\end{equation*}
	namely that $\mu(x,y) = \rho_{X|Y=y}(x)$ as claimed above.

	\bibliographystyle{unsrt}
	\bibliography{bib-b2.bib}

\begin{thebibliography}{10}

\bibitem{LC}
Luca Curcuraci.
\newblock On non-commutativity in quantum theory (i): from classical to quantum
  probability.
\newblock arXiv:1803.04913 [quant-ph], 2018.

\bibitem{khrennikov2009contextual}
Andrei~Y Khrennikov.
\newblock {\em Contextual approach to quantum formalism}, volume 160.
\newblock Springer Science \& Business Media, 2009.

\bibitem{LC3}
Luca Curcuraci.
\newblock On non-commutativity in quantum theory (iii): determinantal point
  process and non-relativistic quantum mechanics.
\newblock arXiv:1803.04921 [quant-ph], 2018.

\bibitem{fairlie1964formulation}
DB~Fairlie.
\newblock The formulation of quantum mechanics in terms of phase space
  functions.
\newblock In {\em Mathematical Proceedings of the Cambridge Philosophical
  Society}, volume~60, pages 581--586. Cambridge University Press, 1964.

\bibitem{baker1958formulation}
George~A Baker~Jr.
\newblock Formulation of quantum mechanics based on the quasi-probability
  distribution induced on phase space.
\newblock {\em Physical Review}, 109(6):2198, 1958.

\bibitem{bunge1970so}
Mario Bunge.
\newblock The so-called fourth indeterminacy relation.
\newblock {\em Canadian journal of physics}, 48(11):1410--1411, 1970.

\bibitem{wang2007introduce}
Zhi-Yong Wang and Cai-Dong Xiong.
\newblock How to introduce time operator.
\newblock {\em Annals of Physics}, 322(10):2304--2314, 2007.

\bibitem{moretti2013spectral}
Valter Moretti.
\newblock {\em Spectral theory and quantum mechanics: with an introduction to
  the algebraic formulation}.
\newblock Springer Science \& Business Media, 2013.

\bibitem{rudnick2004elements}
Joseph Rudnick and George Gaspari.
\newblock {\em Elements of the random walk: an introduction for advanced
  students and researchers}.
\newblock Cambridge University Press, 2004.

\bibitem{oksendal2013stochastic}
Bernt Oksendal.
\newblock {\em Stochastic differential equations: an introduction with
  applications}.
\newblock Springer Science \& Business Media, 2013.

\bibitem{khrennikov2005interference}
Andrei~Yu Khrennikov.
\newblock Interference in the classical probabilistic model and its
  representation in complex hilbert space.
\newblock {\em Physica E: Low-dimensional Systems and Nanostructures},
  29(1-2):226--236, 2005.

\bibitem{maassen1988generalized}
Hans Maassen and Jos~BM Uffink.
\newblock Generalized entropic uncertainty relations.
\newblock {\em Physical Review Letters}, 60(12):1103, 1988.

\bibitem{klebaner2005introduction}
Fima~C Klebaner.
\newblock {\em Introduction to stochastic calculus with applications}.
\newblock World Scientific Publishing Co Inc, 2005.

\bibitem{friesecke2009ehrenfest}
Gero Friesecke and Mario Koppen.
\newblock On the ehrenfest theorem of quantum mechanics.
\newblock {\em Journal of Mathematical Physics}, 50(8):082102, 2009.

\bibitem{wilson1931semi}
Wallace~Alvin Wilson.
\newblock On semi-metric spaces.
\newblock {\em American Journal of Mathematics}, 53(2):361--373, 1931.

\bibitem{pick1899geometrisches}
Georg Pick.
\newblock Geometrisches zur zahlenlehre.
\newblock {\em Sitzenber. Lotos (Prague)}, 19:311--319, 1899.

\bibitem{ehrhart1962geometrie}
Eugene Ehrhart.
\newblock Geometrie diophantienne-sur les polyedres rationnels homothetiques an
  dimensions.
\newblock {\em Comptes Rendus Hebdomadaires Des Seances De L Academie Des
  Sciences}, 254(4):616, 1962.

\bibitem{gudder2014stochastic}
Stanley~P Gudder.
\newblock {\em Stochastic methods in quantum mechanics}.
\newblock Courier Corporation, 2014.

\bibitem{kallenberg2006foundations}
Olav Kallenberg.
\newblock {\em Foundations of modern probability}.
\newblock Springer Science \& Business Media, 2006.

\end{thebibliography}

\end{document}